\documentclass[11pt]{article}

\usepackage{sn-preamble} 
\usepackage{tikz}
\usepackage{verbatim}
\usepackage{cancel}
\usepackage{bm}
\usepackage{caption}
\usepackage{changepage}   
\usepackage{ragged2e}

\usepackage{lscape}
\usepackage{rotating}
\usepackage{placeins}

\newcommand{\violet}[1]{{\color{violet}{#1}}}
\newcommand{\black}[1]{{\color{black}{#1}}}

\newcommand{\ignore}[1]{}

\makeatletter
\newcommand{\vast}{\bBigg@{4}}
\newcommand{\Vast}{\bBigg@{5}}
\makeatother

\newcommand{\qual}{\mathrm{qual}}

\newcommand{\sgn}{\mathrm{sign}}

\newcommand{\xref}{x_{\mathrm{ref}}}
\newcommand{\Adv}{\mathrm{Adv}}

\newcommand{\tbg}{\tilde{\boldsymbol{g}}}

\newcommand{\forster}{\textsc{Forsterize}}

\newcommand{\warmup}{\textsc{Weak-Learn-$\mathsf{AND}_2$-of-$\mathsf{LTF}$}}

\newcommand{\warmupwl}{\textsc{Weak-Learn-$\mathsf{AND}_2$-of-$\mathsf{LTF}$-Over-Distribution}}

\newcommand{\learnfunc}{\textsc{Weak-Learn-$\mathsf{ANY}_k$-of-$\mathsf{LTF}$}}

\newcommand{\nnz}{\mathrm{nnz}}

\newcommand{\imp}{\mathrm{imp}}
\newcommand{\mimp}{\mathrm{mimp}}

\newcommand{\ind}{\mathrm{ind}}

\usepackage{iftex}
\ifPDFTeX
  \usepackage[T1]{fontenc}
  \usepackage{lmodern}
  \newcommand{\crayonfont}{\sffamily} 
\else
  \usepackage{fontspec}
  \newfontfamily\Crayon{TeX Gyre Heros}
  \newcommand{\crayonfont}{\Crayon}
\fi

\usepackage{xcolor}
\usetikzlibrary{calc,fit,positioning,shapes.geometric,decorations.pathmorphing,shadows.blur}

\tikzset{
  crayon draw/.style={
    draw=violet!80!black, line width=1.8pt, rounded corners=10pt,
    decorate, decoration={random steps,segment length=6pt,amplitude=0.9pt}
  },
  crayon line/.style={
    draw=violet!80!black, line width=1.2pt,
    decorate, decoration={random steps,segment length=5pt,amplitude=0.8pt}
  },
  sticker star/.style={
    star, star points=5, star point ratio=1.8,  line join=round, line cap=round,
    minimum size=16pt, inner sep=0pt,
    draw=olive!60!black, line width=0.8pt,
    fill=yellow!85!orange,
    blur shadow={shadow blur steps=6,shadow xshift=1.2pt,shadow yshift=-1.2pt,shadow opacity=0.5}
  },
  title text/.style={font=\crayonfont\Large, text=violet!90!black},
  rank text/.style={font=\crayonfont\large, text=violet!90!black},
  small math/.style={font=\crayonfont\small, text=violet!90!black},
  note/.style={font=\itshape\footnotesize, text=black!85},
}

\title{
Learning Functions of Halfspaces
}
\author{Josh Alman\thanks{Email: \texttt{josh@cs.columbia.edu}}
\\ \textsl{Columbia University} \and 
Shyamal Patel\thanks{Email: \texttt{shyamalpatelb@gmail.com}} \\ \textsl{Columbia University} \and Rocco A. Servedio\thanks{Email: \texttt{ras2105@columbia.edu}}\\ \textsl{Columbia University}}
\date{}

\begin{document}

\pagenumbering{gobble}
\maketitle

\begin{abstract}
We give an algorithm that learns arbitrary Boolean functions of $k$ arbitrary halfspaces over $\R^n$, in the challenging distribution-free Probably Approximately Correct (PAC) learning model, running in time $2^{\sqrt{n} \cdot (\log n)^{O(k)}}$.  This is the first algorithm that can PAC learn even intersections of two halfspaces in time $2^{o(n)}.$
\end{abstract}

\newpage

\pagenumbering{arabic}


\section{Introduction} \label{sec:intro}

Learning an unknown \emph{halfspace}, i.e.,~a Boolean-valued function $f: \R^n \to \bits,$ $f(x) = \sign(w \cdot x - \theta)$, is arguably the most fundamental problem in the theory of machine learning.
Algorithms for learning halfspaces date back more than sixty years, and have long played a central role as the backbone of many different approaches for learning Boolean functions, such as the Perceptron \cite{Novikoff:62,Rosenblatt:58} and Winnow \cite{lit87} algorithms, various statistical approaches \cite{DH73}, Support
Vector Machines \cite{Vap82}, AdaBoost \cite{FS97}, and many others. Such algorithms are also at the core of state-of-the-art learning results for a range of other Boolean function classes, including DNF formulas \cite{KlivansServedio:04jcss}, decision trees \cite{blu92}, and de Morgan Boolean formulas of bounded size \cite{Lee09formulas} (see, e.g.,~the discussion in \cite{HellersteinServedio:07,Sherstov13sicomp,Sherstov13,GKK20}).

In a landmark paper, Blumer et al.~\cite{bluehrhauwar89} showed that the existence of polynomial-time algorithms for linear programming gives a $\poly(n)$-time algorithm to learn a single halfspace in the ``distribution independent'' PAC learning model.  (Recall that in this framework a learning algorithm must succeed with high probability in generating a high-accuracy hypothesis $h: \R^n \to \bits$, given random examples that are  drawn from an unknown and arbitrary distribution $\calD$ over $\R^n$ and labeled by an unknown and arbitrary halfspace; see \Cref{sec:PAC} for a more detailed overview of the PAC learning model.)

Since then, one of the most outstanding open problems in computational learning theory has been to develop non-trivial algorithms for learning more complex functions of halfspaces.  In particular, the problem of learning a simple union or intersection of halfspaces has been the subject of especially intense interest, with multiple papers, from the early 1990s to the present day, highlighting its importance and centrality as a frontier challenge problem for computational learning theory \cite{Baum:90b,KOS:04,KlivansSherstov07unconditional,KS09,KLT09,Vempala10,KhotSaket:11jcss,Tiegel24,KSV24COLT,DMRT25}, with close connections to learning polytopes, multi-index models, and depth-$2$ neural networks.

Given the importance of PAC learning intersections of halfspaces, and the evident difficulty of designing non-trivial algorithms, many researchers have developed algorithms for restricted versions of the problem.  To give a brief and non-exhaustive overview, these include algorithms for learning ``low weight'' halfspaces \cite{KOS:04} as well many algorithms for learning under restricted distributions such as uniform distributions \cite{BlumKannan:97,KOS:04}, Gaussian distributions \cite{Vempala10,Vempalafocs10,KSV24COLT},  logconcave distributions \cite{KLT09,BZ17}, distributions which guarantee that every example has a non-trivial margin \cite{KS08,GKKN22,Chubanov23}, and ``factorizable'' distributions \cite{DMRT25}.  Other algorithms grant more power to the learner, such as the ability to make membership queries in various settings \cite{KwekPitt:98,GKM12}.

A different active line of work, motivated by the difficulty of designing successful algorithms, seeks to establish various kinds of \emph{hardness} results for learning intersections of halfspaces.  This includes complexity-theoretic lower bounds on PAC learning with restricted classes of hypotheses based on worst-case hardness assumptions \cite{BR92,ABFKP:08,KhotSaket:11jcss}; representation-independent lower bounds based on average-case hardness assumptions or cryptographic assumptions \cite{KS09,DSSS16,DV21,Tiegel24,DingGu21};
Statistical Query lower bounds of various types \cite{KlivansSherstov07unconditional,DKPZ21,HSSV22,KSV24COLT}; and lower bounds on state-of-the-art techniques based on polynomial threshold function degree \cite{Sherstov13,Sherstov13sicomp,sherstov2021hardest}.

To summarize the above, a plethora of partial results have been given for a wide range of restricted versions of the intersection-of-halfspaces learning problem, and many hardness results have been established for different versions of the problem. 
But prior to the present work, no $2^{o(n)}$-time algorithm was known for the original problem of PAC learning even an intersection of two halfspaces.

\bigskip

\subsection{Our result}

We give the first non-trivial algorithm for learning intersections (and more general functions) of halfspaces in the general PAC learning model.  To state our result precisely, let ${\cal C}_k$ denote the class of all Boolean functions of $k$ halfspaces,
i.e.,~functions of the form $f: \R^n \to \bits,$ $f(x)=g(h_1(x),\dots,h_k(x))$ where each $h_i: \R^n \to \bits$ is an arbitrary halfspace and $g$ is an arbitrary function $\bits^k \to \bits.$

\begin{theorem} \label{thm:mainintro}
There is an algorithm that runs in time $\poly(2^{\sqrt{n} \cdot (\log n)^{O(k)}},$ $1/\eps,$ $\log(1/\delta))$ and learns ${\cal C}_k$ in the distribution-free Probably Approximately Correct (PAC) learning model, using random examples only. 
\end{theorem}

As stated earlier, this is the first algorithm in the original PAC learning model of learning from random examples that runs in time $2^{o(n)}$, even for learning an intersection of two halfspaces under arbitrary distributions.  In fact, for learning to accuracy $\eps=1/n$, no prior $2^{o(n)}$-time algorithm was known even for learning an intersection of two halfspaces under the uniform distribution on $\bits^n$.

\subsection{Our techniques} \label{sec:techniques}

The standard technique for PAC learning a class ${\cal C}$ of Boolean functions in the distribution-free setting is the so-called ``polynomial method.'' In this approach, one shows that every function $f$ in ${\cal C}$ can be represented exactly as the sign of a real polynomial of degree $d$, i.e.~as a degree-$d$ ``polynomial threshold function'' (PTF). Since we can PAC learn degree-$d$ PTFs in time $\poly({n \choose \leq d})$, this yields a $\poly({n \choose \leq d})$-time algorithm for PAC learning the class ${\cal C}$. Unfortunately, the celebrated work of Sherstov \cite{Sherstov13, Sherstov13sicomp, sherstov2021hardest} proves that even over the domain $\zo^n$, an intersection of two halfspaces may have PTF degree $\Omega(n)$, and hence any $2^{o(n)}$ time algorithm must proceed via different techniques.

Conceptually, the difficulty in learning an intersection $h_1 \wedge h_2$ of even two halfspaces has long been attributed to the ``credit assignment problem'' \cite{Baum:90b}. In particular, when we receive a negative example $x$, we know that either $h_1(x) = -1$ or $h_2(x) = -1$; the difficulty of the problem stems from our inability to distinguish between these two cases.

\subsubsection{Warm-up:  Learning an intersection of two halfspaces.}

To circumvent the above two barriers, we will give an approach to directly solve the credit assignment problem that is not based on polynomials.  Before discussing our approach, we record some initial simplifying assumptions. In particular, we'll assume (without loss of generality, as discussed later) that the halfspaces $h_1$ and $h_2$ are origin-centered, and that all examples are drawn from a distribution ${\cal D}$ supported on the unit sphere $\mathbb{S}^{n-1}$. As such, we have that the labeling function $f(x)$ is $h_1(x) \land h_2(x)$ where $h_1(x) = \sign(w^{(1)} \cdot x)$ and $h_2(x) = \sign(w^{(2)} \cdot x)$ for unknown unit vectors $w^{(1)}, w^{(2)} \in \mathbb{R}^n$. We will also assume that we are given a ``large'' sample $S$ (we will specify its size soon) of points that are drawn from the distribution $\calD$ and labeled by $f$. Our goal will then be to run in $\poly(|S|, 2^{\wt{O}(\sqrt{n})})$ time and find a ``weak hypothesis'' that   correctly classifies $1/2 + \gamma$  fraction of the points in $S$, where $\gamma \geq 2^{-\wt{O}(\sqrt{n})}$.
Moreover, this weak hypothesis will belong to a class with low VC-dimension.  By choosing the sample size $|S|$ to be $\gg 1/\gamma$,  such an algorithm to find such a weak hypothesis over $S$ is sufficient to PAC learn an intersection of two halfspaces to accuracy $\eps$ and confidence $\delta$ in $\poly(2^{\wt{O}(\sqrt{n})}, \eps^{-1}, \log(1/\delta))$ time, using standard generalization bounds and boosting arguments.

The core of our approach will be a procedure to \emph{find a region where the first halfspace is (nearly) constant}. More precisely, we want to either find a region $R_+ \subseteq \mathbb{R}^n$ such that (say) at least a $1 - \frac{1}{\omega(n \log n)}$ fraction
of the points in $S \cap R_+$ satisfy $h_1(x) = 1$, or a region $R_{-} \subseteq \mathbb{R}^n$ such that (say) at least a $1 - \frac{1}{\omega(n \log n)}$  fraction of the points in $S \cap R_{-}$ satisfy $h_1(x) = -1$. In order for such a guarantee to be non-trivial, we will also require that the region we find contain a non-trivial fraction of the points of $S$, i.e., at least a $2^{-\wt{O}(\sqrt{n})}$ fraction.

Note that constructing such a region is sufficient to let us solve the learning problem. Indeed, if we find a region $R_-$, then $f$ must be close to the constant $-1$ function on $R_-$, as almost all points in this region don't satisfy $h_1$. So a hypothesis that outputs $-1$ on $R_-$ and outputs the majority label outside of $R_-$ will have accuracy ${\frac 1 2} + \Omega \left( \frac{|R_- \cap S|}{|S|} \right) = {\frac 1 2} + 2^{-\wt{O}(\sqrt{n})}$ over $S$. On the other hand, if we find a region $R_+$, then all but at most a $\frac{1}{\omega(n \log n)}$ fraction of points in $S \cap R_+$ are labeled according to $h_2$ (because $h_1=1$ and so $h_1 \wedge h_2$ is simply $h_2$). In this case, we can draw $n \log n$ points from $S \cap R_+$ at random, and with $1-o(1)$ probability all of them will be labeled according to $h_2$; we can then use linear programming to learn a high accuracy hypothesis $h$ for $h_2$, with error rate at most $o(1)$ on $S \cap R_+$. Outputting $h(x)$ for points $x \in S \cap R_+$ and the majority label outside of $R_+$, we again  have a hypothesis with accuracy ${\frac 1 2} + \Omega \left( \frac{|R_+ \cap S|}{|S|} \right) = {\frac 1 2} + 2^{-\wt{O}(\sqrt{n})}$ over $S$.

It remains to describe how to construct such a region $R_+$ or $R_-$ in time $\poly(|S|, 2^{\tilde{O}(\sqrt{n})})$. To do this, we will assume that $h_1$ has an $\Omega(1/n)$ fraction of points $x$ with a \emph{non-trivial margin}, i.e., that satisfy $\left| w^{(1)} \cdot x \right| = \Omega(1/\sqrt{n})$. Crucially, such anti-concentration can be achieved by applying a \emph{algorithmic Forster transform} (see \Cref{lem:rad-iso-large-margin} and \Cref{thm:forster}), which was provided in recent work of Diakonikolas et al.~\cite{diakonikolas2023strongly}. We then branch into two cases based on whether $w^{(1)}$ has a ``large positive margin'' or a ``large negative margin,'' meaning that  $\Omega(1/n)$ fraction of points in $S$ either satisfy $w^{(1)} \cdot x \geq \Omega(1/\sqrt{n})$ or satisfy $w^{(1)} \cdot x \leq -\Omega(1/\sqrt{n})$.

{\bf The ``large positive margin'' case:} In this case, we aim to randomly choose a vector $\bg \sim \calN \left(0,\frac{1}{n} I_n \right)$  that satisfies $w^{(1)} \cdot \bg \geq \alpha := \Theta \left( \frac{\sqrt{\log(n)}}{n^{1/4}} \right)$, and we define the region
$R_+$ to be
$R_+(\bg) := \{x: \bg \cdot x \geq \alpha/10\}$.  It can be shown (see \Cref{lem:advantage-simplified}) that points with a ``large positive margin vis-a-vis $w^{(1)}$,'' i.e.~points $x$ satisfying $w^{(1)} \cdot x \geq \Omega(1/\sqrt{n})$, are at least (say) $n^{3}$ times more likely to appear in $R_+(\bg)$ than points $x$ satisfying $h_1(x) = -1$ (equivalently, satisfying $w^{(1)} \cdot x < 0$). Since an $\Omega(1/n)$ fraction of points in $S$ have such a large positive margin vis-a-vis $w^{(1)}$, it follows that we expect at most an $O(n^{-2})$ fraction of the points in $S \cap R_+(\bg)$ to satisfy $h_1(x) = -1$. Moreover, by Gaussian tail bounds, we expect roughly an $\exp(- O(n \alpha^2)) = 2^{-\wt{O}(\sqrt{n})}$ fraction of points to lie in $R_+(\bg)$, as desired.

Of course, we cannot sample a $\bg \sim \calN \left(0,\frac{1}{n} I_n \right)$ that satisfies $w^{(1)} \cdot \bg \geq \alpha$, as we do not know $w^{(1)}$. Instead, we will ``guess'' by randomly drawing a $\bg \sim \calN \left(0,\frac{1}{n} I_n \right)$. Note that by Gaussian tail bounds, such a $\bg$ will indeed satisfy $w^{(1)} \cdot \bg \geq \alpha$ with probability $2^{-\tilde{O}(\sqrt{n})}.$ After $2^{\wt{O}(\sqrt{n})}$ repetitions, with high probability we will indeed guess a $\bg$ satisfying $w^{(1)} \cdot \bg \geq \alpha$, 
Thus, all in all, in this case we indeed construct a region $R_+(\bg)$ in time $2^{\wt{O}(\sqrt{n})}$, as desired.

{\bf The ``large negative margin'' case:}  This case is almost identical. We again draw a guess $\bg \sim \calN \left(0,\frac{1}{n} I_n \right)$ and hope that the vector $\bg$ satisfies $w^{(1)} \cdot \bg \geq \alpha$. The only change is that we consider the region $R_-(\bg) := \{x: \bg \cdot x \leq -\alpha/10\}$ rather than $R_+(\bg)$. The same approach and analysis then gives an algorithm to construct a region $R_-(\bg)$ in time $2^{\wt{O}(\sqrt{n})}$ in the negative margin case, as desired.

\subsubsection{Learning functions of $k$ halfspaces.} \label{sec:learnfunc}

We now turn to the actual (and significantly more challenging) task of learning an arbitrary function of $k$ unknown halfspaces. As before, we can assume that $\calD$ is supported on $\mathbb{S}^{n-1}$ and that each target halfspace is origin centered. Our goal is to learn a target function $f = g(h_1(x), \dots, h_k(x))$, where $h_i(x) = \sign(w^{(i)} \cdot x)$ is the $i$-th unknown halfspace and the combining function $g: \bits^k \to \bits$ is unknown and arbitrary.  Similar to before, we aim to do this by giving a weak learner that finds a hypothesis with non-trivial accuracy (now this will mean accuracy ${\frac 1 2} + \gamma$ where $\gamma = 2^{-\sqrt{n}\log^{O(k)}(n)}$) on a set $S$ of $2^{\sqrt{n} \log^{\Omega(k)}(n)}$ labeled examples.

Inspired by the case of an intersection of two halfspaces, our high-level approach will be to construct a region $R \subseteq \R^n$ such that every halfspace $h_i$ is ``fixed,'' meaning that it almost always takes the same value across all the examples in $S \cap R$ (and moreover $|S \cap R|$ is not too small).  That is, our goal is that there exist values $s_i \in \bits$ such that for all $i \in [k]$, at most (say) a $\frac{1}{n^2}$ fraction of points in $S \cap R$ satisfy  $\sign(w^{(i)} \cdot x) \not = s_i$. Given this, the target function $f$ would be nearly constant across all points in $S \cap R$, taking value $g(s_1, s_2, \dots, s_k)$ on almost all of them. Consequently, the hypothesis that outputs the bit $g(s_1, s_2, \dots, s_k)$ on points in $R$ and the majority label outside of $R$ would have accuracy ${\frac 1 2} + \gamma$ over $S$, where $\gamma = \Omega \left( |S \cap R|/|S| \right)$. 

Let us first discuss finding such a region $R$ when there are $k=2$ halfspaces, as this already captures many of the non-trivial aspects of our techniques for general $k$. 
(For simplicity, throughout the ensuing discussion we do not mention the Forsterization steps that are carried out to ensure a margin each time we sample a halfspace $\bg \sim \calN \left(0,\frac{1}{n} I_n \right)$.) Recall that in the previous section, we only found a region that fixed a single halfspace. To start, we consider the following naive approach: Sample a $\bg^{(1)} \sim \calN \left(0,\frac{1}{n} I_n \right)$ with $w^{(1)} \cdot \bg^{(1)} \geq \alpha$ and compute a region $R_{s_1}(\bg^{(1)})= \{x: \bg^{(1)} \cdot x \geq s_1 \cdot \alpha/10\}$, for some $s_1 \in \{-,+\}$, to fix $h_1$ as described in the previous section. Now using $R_{s_1}(\bg_1) \cap S$ as our new sample of points to classify, 
compute a region $R_{s_2}(\bg^{(2)})$, for some $s_2 \in \{-,+\}$, to fix $h_2$, where $\bg^{(2)} \sim \calN \left(0,\frac{1}{n} I_n \right)$ and $w^{(2)} \cdot \bg^{(2)} \geq \alpha$. Then use $R := R_{s_1}(\bg_1) \cap R_{s_2}(\bg^{(2)})$ as the desired region.

Does this naive approach succeed? Encouragingly, $R$ satisfies some of the properties we want; in particular, it contains a non-trivial ($2^{-\wt{O}(\sqrt{n})}$) fraction of points in $S$ (since each of the two ``fixings'' keeps at least a $2^{-\wt{O}(\sqrt{n})}$ fraction of points), and at most an $O(n^{-2})$ fraction
of points in $S \cap R$ satisfy $h_2(x) \not = s_2$. But there is a major problem, which is that we have no guarantee that $h_1$ is fixed by $R$! In particular, while at most a $O(n^{-2})$ fraction of points in $S \cap R_{s_1}(\bg^{(1)})$ satisfy $h_1(x) \not = s_1$, when we further intersect with $R_{s_2}(\bg^{(2)})$, we are ``zooming in'' on a small ($2^{-\wt{O}(\sqrt{n})}$) fraction of points from $S \cap R_{s_1}(\bg^{(1)})$. As such, it is entirely possible  that most points in $R$ satisfy $h_1(x) \not = s_1$. Indeed, recalling  how these regions are actually constructed, if the points satisfying $h_1(x) \not = s_1$
have large margin with respect to $w^{(2)}$ and those satisfying 
$h_1(x) = s_1$
have a small margin with respect to $w^{(2)}$, then we should very much expect most points in $S \cap R$ to satisfy 
$h_1(x) \not = s_1$.

To circumvent this problem, we consider the following small change to the above procedure: When we make our random guess for $\bg^{(1)}$, we hope for the slightly stronger property that $w^{(1)} \cdot \bg^{(1)} \geq \alpha \log(n)$ (rather than just $\alpha$), and we correspondingly define the region $R_{s_1}^\star(\bg^{(1)}) := \{x: \bg^{(1)} \cdot x \geq s_1 \cdot \alpha \log(n) / 10\}$.  We then restrict to the points in
$R^\star_{s_1}(\bg^{(1)}) \cap S$, and we construct a region $R_{s_2}(\bg^{(2)})$ using a $\bg^{(2)} \sim \calN \left(0,\frac{1}{n} I_n \right)$ satisfying $w^{(2)} \cdot \bg^{(2)} \geq \alpha$ as before. We then take the region $R$ to be $R := R_{s_1}^\star(\bg^{(1)}) \cap R_{s_2}(\bg^{(2)})$.

As before, we still have that $R$ contains a non-trivial ($2^{-\wt{O}(\sqrt{n})}$) fraction of points in $S$, and that at most an $O(n^{-2})$ fraction of points in $S \cap R$ satisfy $h_2(x) \not = s_2$. Notably, because of how we defined $R^\star_{s_1}(\bg^{(1)})$, we now will also have that at most a $n^{-\log(n)}$ fraction of points in $R^\star_{s_1}(\bg^{(1)})$
satisfy $h_1(x) \not = s_1$. As before, this fraction is not enough to ensure that $h_1$ is fixed by $R$. But here is a crucial insight: if $h_1$ is \emph{not} fixed by $R$, then $S \cap R^\star_{s_1}(\bg^{(1)})$ must contain many points with a large margin vis-a-vis $w^{(2)}$. To see this, bucket the points in $S \cap R^\star_{s_1}(\bg^{(1)})$ with $\sign(w^{(1)} \cdot x) \not = s_1$ by their margin with respect to $w^{(2)}$ to get sets $T_\lambda$ consisting of those points with $w^{(2)}$-margin in the range $\left[ \frac{\lambda}{\sqrt{n}}, \frac{2\lambda}{\sqrt{n}} \right)$ for power-of-two values of $\lambda$. 
A careful calculation using \Cref{lem:advantage} and \Cref{lem:monotonicity} shows that a point in $T_\lambda$ is $n^{\Theta(\lambda)}$ times more likely to appear in $R^\star_{s_1}(\bg^{(1)})$ than a typical point $x \in S \cap R_{s_1}^\star(\bg^{(1)})$ with $\sign(w^{(1)} \cdot x) = s_1$.\footnote{Here we are assuming that only a tiny fraction of the points in  $S \cap R_{s_1}^\star(\bg^{(1)})$ with  $\sign(w^{(1)} \cdot x) = s_1$ have a large margin vis-a-vis $w^{(2)}$, as otherwise $S \cap R_{s_1}^\star(\bg^{(1)})$ would already contain many points with a large margin.}
Thus, we conclude that, in expectation, $R$ has at most a $\sum_{\lambda} n^{\Theta(\lambda)} {\frac {|T_\lambda|}{|S \cap R^\star_{s_1}(\bg^{(1)})|}}$ fraction of points satisfying $\sign(w^{(1)} \cdot x) \not = s_1$. 

Since there is at least a $1/n^2$ fraction of points satisfying $\sign(w^{(1)} \cdot x) \not = s_1$ in $R$ and at most a $n^{-\Omega(\log(n))}$ fraction of points in $S \cap R^\star_{s_1}(\bg^{(1)})$ satisfying $\sign(w^{(1)} \cdot x) \not = s_1$, we can infer that that there exists a $T_\lambda$ with ${\frac {|T_\lambda|}{|S \cap R^\star_{s_1}(\bg^{(1)})|}} \geq n^{-O(\lambda)}$ for some $\lambda = \Omega(\log(n))$. By definition of $T_\lambda$, this implies that there exists a $\lambda = \Omega(\log n)$ such that there is a $n^{-O(\lambda)}$ fraction of points in $S \cap R^\star_{s_1}(\bg^{(1)})$ with margin at least $\Omega \left( \frac{\lambda}{\sqrt{n}} \right)$, i.e.,~at least $\Omega \left(\frac{\log(n)}{\sqrt{n}} \right)$, with respect to $w^{(2)}$.
Notably, this is a much larger margin than the $\Omega(1/\sqrt{n})$ margin promised by the Forster transform earlier, although this has come at the cost of only holding for a $n^{-\Omega(\lambda)}$ fraction of the points as opposed to an $\Omega(1/n)$ fraction of the points earlier.

Despite holding only for a smaller fraction of points, having this large margin makes our algorithmic techniques stronger. In particular, if 
$S \cap R^\star_{s_1}(\bg^{(1)})$ has at least an $n^{-\Omega(\lambda)}$ fraction of points with margin $\Omega \left( \frac{\lambda}{\sqrt{n}} \right)$ with respect to $w^{(2)}$,
 then 
 further sampling a new $\bg' \sim \calN \left(0,\frac{1}{n} I_n \right)$ with $w^{(2)} \cdot \bg' \geq \alpha \log(n)$ and constructing $R_{s}(\bg') := \{x: \bg' \cdot x \geq s \cdot \alpha \log(n)/10\}$ for some $s \in \{-,+\}$  yields that at most $n^{-\Omega(\lambda \log(n))}$ fraction of points in $S \cap R^\star_{s_1}(\bg^{(1)})$ satisfy $h_2(x) \not = s$. Notably, this fraction is at most $n^{-\Omega(\log^2(n))}$, as $\lambda \geq \log(n)$ 
 in our above construction, an improvement over the $n^{-\Omega(\log(n))}$ bound that we had for $R^\star_{s_1}(\bg^{(1)})$.
 
This leads us to our final overall algorithm for constructing the desired region $R$ in the case of $k=2$, which roughly proceeds as follows:

\begin{itemize}
    \item Sample a $\bg^{(1)} \sim \calN \left(0,\frac{1}{n} I_n \right)$ with $w^{(1)} \cdot \bg^{(1)} \geq \alpha \log(n),$ and consider the region $R_{s_1}^\star(\bg^{(1)})$ defined earlier, for some $s_1 \in \{-,+\}$. In particular, consider the point set $R_{s_1}^\star(\bg_1) \cap S$ as the new sample $S'$ of points over which we seek a nontrivially accurate weak hypothesis. 

    \item 
If at least an $n^{-\lambda}$ fraction of points in $S'$ have margin $\Omega \left( \frac{\lambda}{\sqrt{n}} \right)$ with respect to $w^{(2)}$ for some $\lambda \geq (\log n)^i$
where $i$ is the recursion depth of the current execution of the algorithm, then recursively
restart the algorithm on $S'$
with $w^{(1)}$ and $w^{(2)}$ swapped, i.e., trying to fix  $w^{(2)}$ first and then $w^{(1)}$.

\item 
Otherwise, guess $\bg^{(2)} \sim \calN \left(0,\frac{1}{n} I_n \right)$ with $w^{(2)} \cdot \bg^{(2)} \geq \alpha$
and define the region $R_{s_2}(\bg^{(2)}) = \{x: \bg^{(2)} \cdot x \geq s_2 \cdot \alpha/10\}$, for some $s_2 \in \{-,+\}$, to fix $h_2$.

\item Output $R := R_{s_1}^\star(\bg^{(1)}) \cap R_{s_2}(\bg^{(2)})$.

\end{itemize}

Note that the algorithm always terminates after at most $\frac{\log(n)}{\log \log(n)}$ recursion depth, as any point has margin at most $1$. It is useful to emphasize that regarding the second and third bullets above, the algorithm has no way of determining whether there exists $n^{-\lambda}$ fraction of points with margin $\Omega \left( \frac{\lambda}{\sqrt{n}} \right)$ in our current point set. As such, we will simply guess whether or not this is the case and hope to get lucky. Since we guess correctly with probability $\frac{1}{2}$ and only make $\frac{\log(n)}{\log \log(n)}$ many such guesses, the cost of these guesses is negligible. In particular, the runtime of the algorithm will be dominated by guessing vectors $\bg$ that have correlation $\alpha \log(n)$ with the target halfspaces $w^{(j)}$. 

Boostrapping our previous analysis, one can show that by repeatedly applying the previous reasoning, on the $i$th recursive restart of the algorithm we will have that the first halfspace that we attempt to fix has at least $n^{-\lambda}$ points with margin $\Omega \left( \frac{\lambda}{\sqrt{n}} \right)$ where $\lambda \geq (\log n)^i$. Moreover, if at recursion depth $i$ we ever do not have $n^{-\lambda}$ points with margin $\Omega \left( \frac{\lambda}{\sqrt{n}} \right)$ for some $\lambda \geq (\log n)^i$, then (cf.~the discussion at the beginning of the ``crucial insight'' earlier) the region $R$ that we output fixes both $h_1$ and $h_2$, as desired.

To fix $k \gg 2$ halfspaces, we follow the same high level idea. At each step, we have fixed some set of halfspaces $F \subseteq [k]$. We then either fix a new halfspace or we increase the margin of some halfspace in $F$ by at least a $\Omega(\log(n))$ factor. The run time is then dominated by the fact that to get such a guarantee, we will need to make a guess $\bg \sim \calN \left(0,\frac{1}{n} I_n \right)$ that satisfies $w^{(1)} \cdot \bg \geq \alpha \log^{O(k)}(n)$. Since we must draw such a $\bg$ by only sampling from  $\calN \left(0,\frac{1}{n} I_n \right)$, we then expect our algorithm to take time $2^{\sqrt{n} \log^{O(k)}(n)}$.

We end by noting that while the above sketch covers the core ideas in our algorithm, it hides a number of details, particularly in the case of $k \gg 2$ halfspaces. To mention one of these, it turns out that our algorithm potentially needs to restart up to $\log^k(n)$ times. Roughly speaking, this arises from the fact that if $h_1, \dots, h_{k/2}$ are fixed and $h_1$ has a large margin, then our algorithmic techniques will allow us to find a region where $h_1$ is very close to constant as above. That said, after restricting to such a region, we may no longer fix $h_2, \dots, h_{k/2}$. While the ideas described above can be modified to continue to ensure that we make progress here, it comes at the cost of an increased number of restarts, and requires us to make guesses $\bg$ that have much bigger correlations, i.e., $\alpha \log^{O(k)}(n)$, with the target halfspaces. 

\medskip


\section{Preliminaries} \label{sec:preliminaries}

\subsection{Basic tools from probability}
We use {\bf bold font} to denote random variables.
We will use the following standard Gaussian tail bound:
\begin{lemma}
[Gaussian tail bound] \label{lem:gauss-lb}
If $\bg \sim \calN(0,1)$, then for all $t > 1$, we have that
\[\Pr_{\bg} [\bg \geq t] = \Theta \left( \frac{1}{t} e^{-t^2/2} \right). \]
\end{lemma}

We will also use the following simple ``reverse Markov'' inequality:
\begin{lemma}
[Reverse Markov] \label{lem:reverse-Markov}
Let $\kappa, u > 0$, and let $\bz$ be a real random variable that always takes values less than $u$ and that has $\E[\bz] \geq \kappa u$.  Then
$\Pr[\bz \geq {\frac \kappa 2}u] \geq {\frac \kappa 2}.$
\end{lemma}

\subsection{Origin-Centered Halfspaces}
\label{sec:OCH}
Recall that a \emph{halfspace} is a Boolean-valued function $f: \R^n \to \{-1,1\}$ of the form $f(x)=\sign(w \cdot x - \theta),$ for some $w \in \R^n$ and $\theta \in \R$, where
\[
\sign(t) := \begin{cases}
    1 & 
    \text{if~}t \geq 0\\
    -1 & \text{if~} t<0
\end{cases}
\]
and we view $-1$ as False, 1 as True (so the intersection of two halfspaces $h_1,h_2$ is equal to 1 if and only if $h_1=h_2=1$). We call $w$ the \emph{weight vector} and $\theta$ the \emph{threshold}.

Throughout the paper we will assume that all halfspaces are origin-centered, i.e.,~of the form $\sign(w \cdot x)$ with threshold zero.   It is a standard fact that this is without loss of generality, since we can view any $n$-dimensional example $x$ as an $(n+1)$-dimensional example $x' = (x,1)$, and any $n$-dimensional halfspace $h(x) = \sign(w \cdot x - \theta)$ as an $(n+1)$-dimensional origin-centered halfspace $h'(x') = \sign(w' \cdot x')$ with $w'=(w,-\theta)$, which satisfy $h'(x') = h(x)$. Working with origin-centered halfspaces will facilitate our use of the Forster Transform (see \Cref{remark:Forster}).

Since scaling the weight vector of an origin-centered halfspace does not change the label of any point, we will assume that $w$ is always a unit vector.

We will also assume without loss of generality that all of the halfspaces $h^{(i)}(x)=\sign(w^{(i)} \cdot x)$ in the target function $f=g(h^{(1)},\dots,h^{(k)})$ satisfy $w^{(i)} \cdot x \not = 0$ for any example $x$ in the finite sample drawn from the distribution that is used by our learning algorithm. A simple argument shows that any halfspace can be slightly perturbed to achieve this without changing its label of any example from our finite sample.

Finally, throughout the paper we view the dimension $n$ as an asymptotically large parameter.

\subsection{Margin and Radial Isotropic Position}

We begin by defining the margin of a point vis-a-vis a halfspace:

\begin{definition}[Margin]
    Given a halfspace $\sgn(w \cdot x)$ for some unit vector $w \in \R^{n}$, and given a unit vector $x \in \R^{n}$, the \emph{margin} of $x$ is defined as $|w\cdot x|$.
\end{definition} 

Typically, having examples with a large margin makes learning halfspaces easier. Indeed, if all points have large margin, then well-known approaches such as the Perceptron algorithm \cite{Novikoff:62,Rosenblatt:58} give fast learning algorithms. While we are unable to guarantee that all the examples we draw have large margin, we can ensure that \emph{some points} have large margin. 
More precisely, we will apply an algorithmic Forster transform (which we describe next) to our data set to ensure it satisfies \Cref{def:approx-RIP} below.  (Recall that for two $n \times n$ matrices $A,B$, we write ``$A\preceq B$'' to indicate precedence in the Loewner order, i.e.,~that $B-A$ is positive semi-definite, meaning that $x^T(B-A)x \geq 0$ for every nonzero $x \in \R^n$.)

\begin{definition}[$(1 + \eps)$-Radial Isotropic Position with respect to $V$] \label{def:approx-RIP}
    For $\eps \in [0,1/2]$ and a subspace $V \subseteq \R^n$, we say that a set of points $S \subseteq V$ is in \emph{$(1 + \eps)$-radial isotropic position with respect to $V$} if $\|x\|_2 = 1$ for all $x \in S$ and
        \[\frac{1 - \eps}{\dim(V)} \cdot P_V \preceq \frac{1}{|S|} \sum_{x \in S}  xx^T \preceq \frac{1 + \eps}{\dim(V)} \cdot P_V, \]
        where we write $P_V: \R^n \to \R^n$ to denote the orthogonal projection operator onto $V$.
\end{definition}

Crucially, this will imply that a non-trivial fraction of points have large margin for each of the halfspaces we are dealing with. As much of our analysis will deal with the fraction of points having a given margin, we make the following central definition.

\begin{definition} \label{def:margin}
Let $w \in \R^n$ be a unit vector.
We say that a finite data set $S$ of unit vectors has a \emph{$(p,\tau)$-margin with respect to $w$} if
at least $p|S|$ of the $|S|$ points $z \in S$ satisfy $|w \cdot z| \geq \tau.$
\end{definition}

We then have that

\begin{lemma} \label{lem:rad-iso-large-margin}
    If finite set $S \subseteq \R^n$ is in $(3/2)$-radial isotropic position with respect to $V$, then for any unit vector $w \in V$, the set $S$ has a $\left (\frac{1}{4n}, \frac{1}{2\sqrt{n}} \right)$-margin with respect to $w$. 
\end{lemma}

\begin{proof}
Fix any unit vector $w \in V$. Note that 
        \[\Ex_{\bx \sim S} [(w^T \bx)^2] = w^T \Ex_{\bx \sim S} [ \bx \bx^T] w \geq \frac{1}{2\dim(V)} w^T P_V w = \frac{1}{2\dim(V)}. \]
    Since $\|x\| \leq 1$, we must have $|w\cdot x| \leq 1$. Thus, 
        \[\Ex_{\bx \sim S} [(w^T \bx)^2] \leq \Prx_{\bx \sim S} \left[|w^T \bx| \geq \frac{1}{2 \sqrt{\dim(V)}} \right] + \left( \frac{1}{2 \sqrt{\dim(V)}} \right)^2.\]
    Rearranging and using that $\dim(V) \leq n$ then gives the desired result.
\end{proof}

Given a finite set of points $S \subseteq \R^n$ where $0^n \notin S$ we can always efficiently put a significant fraction of them into radial isotropic position with respect to some subspace $V$ by using an algorithm of Diakonikolas, Tzamos, and Kane:

\begin{theorem}[Algorithmic Forster Transform \cite{diakonikolas2023strongly}\footnote{We remark that this is a slight rephrasing of the result in \cite{diakonikolas2023strongly}. In particular, they set $A: V \rightarrow \mathbb{R}^{\dim(V)}$ and guarantee that $S'$ in $(ii)$ is in $(1+\eps)$-radial isotropic position with respect to $\R^{\dim(V)}$. This is equivalent to our formulation as we can simply apply an isometry from $\R^{\dim(V)}$ to $V$. Having $A$ map from $V$ to $V$ will be convenient in our algorithm and analysis below as it keeps us in the same ambient space $\mathbb{R}^n$.}]
\label{thm:forster}
There exists a randomized algorithm \forster\ that, given a multi-set $S \subseteq \R^n \setminus \{0\}$ and $\eps \in (0,1)$, runs in time $\poly(|S|n/\eps)$ and with high probability returns a subspace $V \subseteq \R^n$ with $V \not  = 0$, a linear transformation $A: V \rightarrow V$, and a set of points $S' \subseteq V$ such that

\begin{itemize}

    \item [$(i)$] $|S \cap V| \geq |S| \cdot {\frac {\dim(V)} n}$, and
 
    \item [$(ii)$] $S' := \left\{ \frac{Ax}{\|Ax\|}: x \in S \cap V \right\} \subseteq V$ is in $(1 + \eps)$-radial isotropic position with respect to $V$.
    
\end{itemize}
 
\end{theorem}

Throughout the paper, we write ``\forster'' to denote an invocation of the \forster\ algorithm with its $\eps$-parameter set to 1/2.

Note that the \forster\ algorithm, strictly speaking, takes as input a multi-set of unlabeled examples. When there is no risk of confusion we sometimes write \forster$(S)$, where $S$ is a multi-set of labeled examples $\{(x^{(1)},y^{(1)}),\dots,(x^{(|S|)},y^{(|S|)})\}$, to mean $\forster(\{x^{(1)},\dots,x^{(|S|)}\}$, and we view the output $S'$ of \forster$(S)$ in such cases as being the corresponding set of labeled examples (see \Cref{remark:Forster}).

\begin{remark} [The Forster transform is compatible with origin-centered halfspaces]
\label{remark:Forster}
We observe that if $h^{(i)}(x) = \sign(w^{(i)}\cdot x)$ is any one of the $k$ halfspaces over $\R^n$ in the target function and $S,V,A,S'$ are as in \Cref{thm:forster}, then for all nonzero vectors $x \in V$ we have that $h^{(i)}(x) = \sign(w' \cdot {\frac {Ax}{\|Ax\|}})$, where
$w' = \frac{A^{-1}P_V w^{(i)}}{\|A^{-1}P_V w^{(i)}\|} \in V$. (Note this will always be well-defined since $w^{(i)}\cdot x \not = 0$ for all $x $ in $S$ by assumption as discussed near the end of \Cref{sec:OCH}.)
In words, the property of being labeled by an origin-centered halfspace (or by a function of origin-centered halfspaces) is ``preserved under Forsterization.''
\end{remark}


\section{Warm-Up: Learning an Intersection of Two Halfspaces} \label{sec:warmup}

In this section we prove the following theorem:

\begin{theorem} \label{thm:2}
There is a distribution-free PAC learning algorithm for the class of all intersections of two halfspaces over $\R^n$, which learns to accuracy $\eps$ and confidence $1-\delta$, running in time $\poly(2^{\tilde{O}(\sqrt{n})},1/\eps,\log(1/\delta))$ and using $\poly(2^{\tilde{O}(\sqrt{n})},1/\eps,\log(1/\delta))$ examples.
\end{theorem}

\Cref{thm:2} is easily seen to be a special case of \Cref{thm:mainintro}; we give a self-contained proof of \Cref{thm:2} in this section as a warm-up, and to highlight the ideas which suffice for learning an intersection of two halfspaces.  (As discussed in \Cref{sec:techniques}, additional conceptual and technical ingredients are required for the more ambitious goal of learning arbitrary functions of $k>2$ halfspaces.)

The key to \Cref{thm:2} will be the \warmup{}
 algorithm that is given in \Cref{alg:weak-learn-two} below.
At an intuitive level, lines~2(a)-2(c) of the algorithm attempt to generate a hypothesis which ``abstains'' outside the region $R_-(\bg)$ and outputs the constant $-1$ within that region, while lines~2(d)-2(e) try to generate a hypothesis which ``abstains'' outside the region $R_+(\bg)$ and predicts according to a linear threshold function within that region.
Our analysis will show that each time through the Line~2 loop, there is at least a $2^{-\tilde{O}(\sqrt{n})}$ probability that either Lines~2(a)-2(c) or Lines~2(d)-2(e) succeed in constructing a hypothesis with accuracy at least ${\frac 1 2} + 2^{-\tilde{O}(\sqrt{n})}$ on the input data set $S$ of labeled examples; this is established in the proof of \Cref{lem:weak-learn-over-sample}, which is the main result of \Cref{sec:advantage}. Once we have shown this, it is straightforward to use \warmup{} to obtain a distribution-independent PAC learning algorithm via standard generalization error arguments and accuracy boosting methods, thereby proving \Cref{thm:2} (see \Cref{ap:proof-theorem-2} for details of how this is done).

{\small 
\begin{algorithm}
\addtolength\linewidth{-2em}
\vspace{0.5em}

\textbf{Input:} A sequence $S$ of $2^{\Omega(\sqrt{n} \log (n))}$ labeled examples 
$(x,y)$
where each 
$x \in \R^n$ 
and each 
$y=f(x)$ for an unknown $f: \R^n \to \bits$ that is an intersection of two halfspaces, 
$f(x) = \sgn(w^{(1)} \cdot x) \land \sgn(w^{(2)} \cdot x)$.\\[0.25em]
\textbf{Output:} A hypothesis $h: \R^n \to \bits$ or FAIL.\\[0.5em]

\warmup($S$):\\

\begin{enumerate}
	\item \vskip.05in Let $V,A,S' \gets \forster(S)$.

        \item \vskip -.05in Set $\beta := \frac{\sqrt{\ln(n)}}{n^{1/4}}$, and 
        repeat the following $2^{\widetilde{O}(\sqrt{n})}$ times: 
        \vskip -0.05in
        \label{line:warmup-loop}
      
        \begin{itemize}
            \item [(a)] Sample $\bg \in \R^n$ from the distribution $\calN(0,\frac{1}{n} I_n)$.
            \item [(b)] 
            Let $R_-(\bg) := \{x \in \R^n: \bg \cdot x \leq -\beta \}$.
            \item [(c)] 
            For each $b_1,b_2 \in \bits,$
            define $\bh_{b_1,b_2,-}: \R^n \to \bits$ as follows:  
                \[\bh_{b_1,b_2,-}(x) := \begin{cases} b_1 & \text{if~}x \not \in V \\ 
                b_2 & \text{if~}x \in V \text{~and~} \frac{Ax}{\|Ax\|} \not \in R_-(\bg) 
                \\ -1 & \text{if~} x \in V \text{~and~} \frac{Ax}{\|Ax\|} \in R_-(\bg) 
                \end{cases}.\]
            If some $\bh_{b_1,b_2,-}$ agrees with at least $1/2 + 2^{-O(\sqrt{n\ln(n)})}$ of the examples in  $S$, output that $\bh_{b_1,b_2,-}$ and halt.
            \end{itemize}

            \begin{itemize}
            \item [(d)] Let $R_+(\bg) := \{x \in \R^n: \bg \cdot x \geq \beta \}$.
            \item [(e)] If $|R_+(\bg) \cap S'| \geq 2^{-\Omega(\sqrt{n})} |S'|$:
                \begin{itemize}
                    \item [(i)] Sample $n\ln(n)$ points $\bx^{(1)}, \dots, \bx^{(n\ln(n))}$ uniformly and independently from $R_+(\bg) \cap S'$.
                    \item [(ii)] Via linear programming,  find a $w \in \R^m$ such that $\sign(w \cdot \bx^{(i)}) = f(\bx^{(i)})$ for all $i \in [n\ln(n)]$ if such a $w$ exists. 
                    \item [(iii)] If such a $w$ exists, 
                    for each $b_1,b_2 \in \bits,$ define $\bh_{b_1,b_2,+}: \R^n \to \bits$ as follows: 
                        \[\bh_{b_1,b_2,+}(x) := \begin{cases} b_1 & \text{if~}x \not \in V \\
                        b_2 & \text{if~}x \in V \text{~and~}\frac{Ax}{\|Ax\|} \not \in R_+(\bg) 
                        \\
                        \sgn \left(w \cdot Ax \right) & \text{if~}x \in V \text{~and~}\frac{Ax}{\|Ax\|} \in R_+(\bg) 
                        \end{cases}.\]
                        If some $\bh_{b_1,b_2,+}$ agrees with at least $1/2 + 2^{-O(\sqrt{n}\ln(n))}$ of the examples in  $S$, output that $\bh_{b_1,b_2,+}$ and halt.
                \end{itemize}
        \end{itemize}
    \item \vskip-.05in Return FAIL.
\end{enumerate}
\caption{A Weak Learner for Intersections of Two Halfspaces Over a Sample}
\label{alg:weak-learn-two}
\end{algorithm}
}

The rough intuition for the algorithm is that we wish to use our guess $\bg$ to restrict our attention either 

\begin{enumerate}[label=(\Alph*), ref=\Alph*]
  \item\label{item:A} to a region $R_-(\bg)$ where almost all points in $S'$ \emph{don't} satisfy $w^{(1)}$, or
  \item\label{item:B} to a region $R_+(\bg)$ where almost all points in $S'$ \emph{do} satisfy $w^{(1)}$.
\end{enumerate}

Note that in the first case, $f$ should label almost all points in $R_-(\bg) \cap S'$ negatively, since $\sign(w^{(1)} \cdot x)$ is almost never satisfied. On the other hand, in the second case almost all points in $R_+(\bg) \cap S'$ are labeled according to \emph{just one} halfspace, $\sign(w^{(2)} \cdot x)$. Since there is (almost) no noise, sampling a set of $n\ln(n)$ points and finding a consistent halfspace should then have high accuracy over $R_+(\bg) \cap S'$.

A high-level overview of why the \forster\ algorithm achieves our wish is the following: we first apply a Forster transform to ensure that (i) the size of $S'$ is a non-trivial fraction (at least $1/n$) of the size of $S$, and (ii) moreover, by \Cref{lem:rad-iso-large-margin}, a non-trivial fraction of the points in $S'$ have a large margin with respect to $w^{(1)}$. 
We will repeatedly guess random Gaussian vectors $\bg$, in the hope that we find one with $\bg \cdot w^{(1)} \gtrsim n^{-1/4}$ (note that the probability that a random $\bg$ has this property is $\gtrsim 2^{-\sqrt{n}}$). 
We will show that points with a large positive margin with respect to $w^{(1)}$ are much more likely to appear in $R_+(\bg)$ than points with $w^{(1)} \cdot x \leq 0$ --- this is the upshot of \Cref{lem:advantage-simplified}, which is the main lemma that we use. Thus, assuming our guess is good, we indeed expect $R_+(\bg) \cap S'$ to mostly contain points that satisfy $w^{(1)}$. An analogous argument explains why $R_-(\bg) \cap S'$ should mostly consist of points that do not satisfy $w^{(1)}$.

\FloatBarrier

\subsection{Achieving Nontrivial Accuracy on the Input Data set $S$} \label{sec:advantage}

\begin{definition}[Lucky vector]
    We say that a vector $g \in \R^n$ is \emph{$\alpha$-lucky with respect to unit vector $w \in \R^n$} if $w \cdot g \geq \alpha$.
\end{definition}

Throughout \Cref{sec:warmup} we will fix the parameters
\begin{equation}
    \label{eq:alpha-beta-warmup}
    \alpha := 10 \beta, 
    \quad \quad \beta := {\frac {\sqrt{\ln n}}{n^{1/4}}} \quad \text{as in \Cref{alg:weak-learn-two}}.
\end{equation}
Our analysis will condition on $\bg$ being $\alpha$-lucky with respect to $w' := {\frac {A^{-1}P_V w^{(1)}}{\|A^{-1} P_V w^{(1)}\|}}$, 
where $\bg \sim \calN(0,\frac{1}{n} I_n)$ as in Line~3(a). Note that, by \Cref{lem:gauss-lb}, $\bg$ is $\alpha$-lucky with respect to $w'$ with probability $2^{-\wt{O}(\sqrt{n})}$. Since we run the loop in Line \ref{line:warmup-loop} for $2^{\wt{\Omega}(\sqrt{n})}$ times, we expect to sample many $\alpha$-lucky guesses $\bg$.

The next definition captures the idea that a random region $\bR$ of $\R^n$ is more likely to contain one particular unit vector, $x,$ than another one, $\xref$. (All of our analysis in this section deals with unit vectors $x,\xref \in \R^n$; the vectors $x,\xref$ should be thought of as points in $S'$.)

\begin{definition} \label{def:advantage}
Fix $x,\xref$ to be unit vectors in $\R^n$ and let $\bR$ be a random variable which takes values that are regions of $\R^n$.   We say that the
\emph{advantage of $x$ over $\xref$ in $\bR$} is
\[
\Adv(x,\xref,\bR) := {\frac {\Pr[x \in \bR]} {\Pr[\xref \in \bR]}}.
\]
(We will usually consider advantages that are $\gg 1$.)
\end{definition}

The random regions that we will be interested in are the following:

\begin{definition} \label{def:Rsets}
Given a
(not necessarily Gaussian) random vector $\bg \in \R^n$, define random variables $R_{-}(\bg), R_+(\bg) \subset \R^n$ as in \Cref{alg:weak-learn-two}, i.e.,
\begin{align*}
R_{-}(\bg) &:=   \{x \in \R^n: \bg \cdot x \leq - \beta\},\\
R_+(\bg) &:=   \{x \in \R^n: \bg \cdot x \geq \beta\}.
\end{align*}
\end{definition}

Note that when we combine \Cref{def:advantage} with \Cref{def:Rsets}, we get that
\[
\Adv(x,\xref,R_+(\bg))  = \frac{\Pr[x \in R_{+}(\bg)]}{\Pr[\xref \in R_{+}(\bg)]}
= {\frac {\Pr_{\bg}[\bg \cdot x \geq \beta]}{\Pr_{\bg}[\bg \cdot \xref \geq \beta]}}
\]
and likewise for $\Adv(x,\xref,R_-(\bg))$; we will use this often below.

\subsubsection{The Advantage Lemma} 

The following lemma is our main tool to prove \Cref{thm:2}, our main result for learning an intersection of two halfspaces.  Intuitively, it says that if $\bg$ is $\alpha$-lucky for a unit vector $w$, then a point $x$ with a ``large positive margin'' w.r.t.~$w$ is much more likely to be in the region $R_+(\bg)$ than a point $\xref$ that has a non-positive margin w.r.t.~$w$.
\begin{lemma}[Advantage Lemma, Simplified Version for Intersection of Two Halfspaces] \label{lem:advantage-simplified}
Let $x,\xref, w$ be unit vectors in $\R^n$ satisfying $w \cdot \xref \leq 0$ and $w \cdot x \geq {\frac 1 {2\sqrt{n}}}$.
Let $\bg \sim \calN(0, \frac{1}{n} I_n)$ and let $\tbg$ denote $\bg$ conditioned on it being $\alpha$-lucky for $w$.  Then\footnote{Note that the particular constant $7/3$ is not significant and we could have obtained a range of different constants here; $7/3$ was chosen because it is convenient for our later arguments to have here a constant strictly greater than 2.}
\[
\Adv(x,\xref, R_{+}(\tbg)) \geq n^{7/3}.
\]
\end{lemma}

\begin{proof}
Without loss of generality let $w = e_1$ and note that as a consequence of this convention and our assumptions, we have that $\tbg_1 \geq \alpha$, $(x_{\mathrm{ref}})_1 \leq 0,$ and $x_1 \geq {\frac 1 {2\sqrt{n}}}$.
    Fix some $\gamma \geq \alpha$ and
    condition on $\tbg_1 = \gamma$. It then follows that $(\tbg_2, \dots, \tbg_n)$ is drawn according to $\calN(0,\frac{1}{n} I_{n-1})$. Now define the function $t_\gamma: \R^n \to \R,$
    \begin{equation}
        \label{eq:tdef}
    t_{\gamma}(z) := \frac{\beta - \gamma z_1}{\sqrt{1 - z_1^2}} \cdot \sqrt{n},
    \quad \quad \text{so} \quad \quad
    t_{\gamma}(\xref) \geq \beta \sqrt{n}.
    \end{equation}
Note that since $(\tbg_2, \dots, \tbg_n) \sim \calN(0,\frac{1}{n} I_{n-1})$, it follows that
\begin{equation}\Prx_{\tbg} \left[x \in R_+(\tbg) \bigg |\, \tbg_1 = \gamma \right] = \Prx_{\by \sim \calN(0,1)} [\by \geq t_\gamma(x)] = \begin{cases} \Theta(\frac{1}{t_{\gamma}(x)} e^{-t_{\gamma}(x)^2/2}) & \text{if~}t_{\gamma} (x)\geq 1 \\ \Theta(1) & \text{otherwise} \end{cases},
\label{eq:carrot-simple}
\end{equation} 
and likewise
\begin{equation}\Prx_{\tbg} \left[\xref \in R_+(\tbg)  \right] = \Prx_{\by \sim \calN(0,1)} [\by \geq t_\gamma(\xref)] = \Theta\pbra{\frac{1}{\beta \sqrt{n}} e^{-n\beta^2/2}} =
\Theta\pbra{
{\frac {e^{- \sqrt{n} \ln(n)/2}} {n^{1/4} \sqrt{\ln n}}}
},
\label{eq:beet-simple}
\end{equation} 
where in both cases we used \Cref{lem:gauss-lb} (note that in \Cref{eq:beet-simple} we know that $t_\gamma(\xref) \geq \beta \sqrt{n}$ is $>1$ by assumption on $\beta$).

Since the Gaussian tail $\Pr_{\by \sim {\cal N}(0,1)}[\by \geq t]$ is a decreasing function of $t$ and $t_\gamma(x)$ is decreasing in $\gamma$, the probability given in \Cref{eq:carrot-simple} is increasing in $\gamma$. 
Thus, for each $\gamma \geq \alpha$, we have that

\[\Prx_{\tbg} \left[x \in R_+(\tbg) \bigg | \tbg_1 = \gamma \right] \geq \begin{cases} \Omega\pbra{\frac{1}{t_{\alpha}(x)} e^{-t_{\alpha}(x)^2/2}} & \text{if~}t_{\alpha}(x) \geq 1 \\ \Theta(1) & \text{otherwise} \end{cases}. \]
Since $\tbg$ is a mixture over outcomes with $\tbg_1=\gamma$ as $\gamma$ ranges over $[\alpha,\infty)$, we have
\begin{equation} \label{eq:parsnip-simple}
\Prx_{\tbg} \left[x \in R_+(\tbg)  \right] \geq \begin{cases} \Omega\pbra{\frac{1}{t_{\alpha}(x)} e^{-t_{\alpha}(x)^2/2}} & \text{if~}t_{\alpha}(x) \geq 1 \\ \Theta(1) & \text{otherwise} \end{cases}.
\end{equation}
We now consider two cases depending on the value of $t_\alpha(x).$

The first case is that $t_{\alpha}(x) < 1$; in this case, recalling \Cref{def:advantage}, \Cref{eq:beet-simple,eq:parsnip-simple} give us that
\[ \Adv(x,\xref, R_{+}(\tbg)) 
= \Theta\pbra{
e^{\sqrt{n} \ln(n)/2} \cdot n^{1/4} \ln n
} \geq n^{7/3}, 
\]
where the inequality holds recalling that $n$ is an asymptotically large parameter.

The second case is that $t_{\alpha}(x) \geq 1$; note that this implies ${\frac 1 {2 \sqrt{n}}} \leq x_1 \leq 1/10.$ In this case, again by \Cref{eq:beet-simple,eq:parsnip-simple}
we have 
\begin{equation}
\Adv(x,\xref, R_{+}(\tbg)) = \Omega \left( \frac{n^{1/4} \sqrt{\ln(n)}}{t_\alpha(x)} e^{-t_{\alpha}(x)^2/2 + \sqrt{n}\ln(n)/2}\right).
\label{eq:sugakix-simple}
\end{equation}
Recalling \Cref{eq:tdef}, we can then compute that
\begin{align*}
{\frac {-t_\alpha(x)^2} 2} + {\frac {\sqrt{n}\ln(n)}2} 
&= {\frac {n} {1-x_1^2}}
\pbra{
{\frac {-\beta^2 + 2\alpha \beta x_1 - \alpha^2 x_1^2} 2}
} 
+ {\frac {\sqrt{n} \ln(n)}{2}}\\
&={\frac {\sqrt{n}} {2(1-x_1^2)}}\cdot
\pbra{-\ln(n) + 20 \ln(n)x_1 - 100\ln(n)x_1^2 + \ln(n)(1-x_1^2)
}
\tag{recalling \Cref{eq:alpha-beta-warmup}}\\
&= 
{\frac {\sqrt{n}} {2(1-x_1^2)}}\cdot
\pbra{20 \ln(n)x_1 - 101\ln(n)x_1^2}\\
&\geq 2.475\ln(n)
\tag{since ${\frac 1 {2 \sqrt{n}}} \leq x_1 \leq 1/10$}.
\end{align*}
Since $\beta - \alpha x_1 \leq \sqrt{\ln(n)}/n^{1/4}$ and $0 \leq x_1 \leq 1/10,$ it follows from \Cref{eq:tdef} that $t_\alpha(x) \leq O(n^{1/4}\sqrt{\ln(n)} )$, so combining with \Cref{eq:sugakix-simple} and recalling again that $n$ is asymptotically large, we get that 
\[
\Adv(x,\xref, R_{+}(\tbg))
= \Omega(n^{2.475}) \geq n^{7/3}.
\qedhere
\]
\end{proof}

\subsubsection{Filtering}

Our next result, \Cref{lem:filtering-simple}, is a ``filtering lemma.'' 
Before stating the lemma we first give an intuitive explanation:  Let $T$ be a collection of unit-vector example points such that a not-too-tiny fraction of them (at least a $p$ fraction) have a not-too-small margin (at least a $\tau$ margin) with respect to a unit vector $w.$  Let $\bg  \sim {\cal N}(0,{\frac 1 n}I_n)$ be a Gaussian ``guess vector,'' and recall that 

\begin{itemize}

\item $T \cap R_+(\bg)$ is the subset of points in $T$  that have a ``$\beta$-magnitude positive margin under $\bg$,'' 

\item $T \cap R_-(\bg)$ is the subset of points in $T$  that have a ``$\beta$-magnitude negative margin under $\bg$.''

\end{itemize}

\noindent 
\Cref{lem:filtering-simple} says that then with not-too-tiny probability over $\bg$, either

\begin{enumerate}

\item both (i) the fraction of points in $T \cap R_+(\bg)$ that are classified \emph{negatively} by the halfspace $\sign(w \cdot x)$ is small, and (ii) the fraction of points in $T$ that belong to $R_+(\bg)$ is not too small; or

\item both (i) the fraction of points in $T \cap R_-(\bg)$ that are classified \emph{positively} by the halfspace $\sign(w \cdot x)$ is small, and (ii) the fraction of points in $T$ that belong to $R_-(\bg)$ is not too small.

\end{enumerate}

We will eventually apply this filtering lemma by taking $w$ to be $w' = \frac{A^{-1}P_V w^{(1)}}{\|A^{-1}P_V w^{(1)}\|} \in V$, the vector corresponding to the first of the two halfspaces in the target, and taking $T$ to be the set $S'$. The filtering lemma is useful for us for the following reason:
\begin{itemize}

\item In case (1.) of the lemma, restricting our attention to $R_+(\bg)$ means that we have effectively ``filtered'' the initial data set $T$ to the not-too-small subset $T \cap R_+(\bg)$ of points which are essentially labeled \emph{positively} by the first halfspace $\sign(w^{(1)} \cdot x)$, and hence essentially labeled by the target intersection of halfspaces $f$ according to $\sign(w^{(2)} \cdot x)$ (recall Case (\ref{item:B}) in the discussion at the start of \Cref{sec:warmup}). Weak learning over $T \cap R_+(\bg)$ is straightforward in this case, by just trying to learn a linear threshold function (corresponding to $\sign(w^{(2)} \cdot x)$).  

\item
In case (2.) of the lemma, similar to case (1.), restricting our attention to $R_-(\bg)$ means that we have effectively ``filtered'' the initial data set $T$  to the not-too-small subset $T \cap R_-(\bg)$ of points which are essentially labeled \emph{negatively} by the first halfspace $\sign(w^{(1)} \cdot x)$, and hence are essentially labeled \emph{negatively} overall by $f$ (recall case (\ref{item:A}) in the discussion at the start of \Cref{sec:warmup}). In case (2.) weak learning over $T \cap R_-(\bg)$ is even more straightforward than in case (1.), since the constant $-1$ classifier suffices.

\end{itemize}

We remark that while the detailed statement of the lemma is slightly technical, the argument establishing it is quite simple using the advantage lemma (\Cref{lem:advantage-simplified}).

\begin{lemma}[Filtering Lemma] \label{lem:filtering-simple}
    Let $T$ be a finite set of unit vectors in $\R^n$ and let $w$ be a unit vector in $\R^n$ such that $T$ has a $(p,\tau)$ margin with respect to $w$, where $p = {\frac 1 {4n}},\tau={\frac 1 {2 \sqrt{n}}}.$  Then either
    \begin{align}
&        \Prx_{\bg \sim \calN(0,\frac{1}{n}I_n)} \Bigg[ \bigg|\{x \in T \cap R_{+}(\bg): w \cdot x \leq 0\} \bigg| \leq \frac{16}{n^{4/3}} |T \cap R_{+}(\bg)| 
\nonumber \\ &\text{~~~~~~~~~~~~~~~~~and } 
        |T \cap R_{+}(\bg)| \geq 2^{-O(\ln(n) \cdot \sqrt{n})}|T| \Bigg] \geq 2^{-O(\sqrt{n} \ln(n))} 
        \label{eq:first}
    \end{align}
    or 
    \begin{align}
&        \Prx_{\bg \sim \calN(0,\frac{1}{n}I_n)} \Bigg[ \bigg|\{x \in T \cap R_{-}(\bg): w \cdot x \geq 0\} \bigg| \leq \frac{16}{n^{4/3}} |T \cap R_{-}(\bg)|  \nonumber \\ &\text{~~~~~~~~~~~~~~~~~and }        |T \cap R_{-}(\bg)| \geq 
2^{-O(\ln(n) \cdot \sqrt{n})}|T| \Bigg] \geq 2^{-O(\sqrt{n} \ln(n))}.
	\label{eq:second}
    \end{align}
\end{lemma}
\begin{proof}
Since $T$ has a $(p,\tau)$ margin with respect to $w$, either at least $p|T|/2$ points $x \in T$ have $w \cdot x \geq \tau$, or at least $p|T|/2$ points $x \in T$ have $w \cdot x \leq -\tau.$  Swapping $w$ and $-w$ swaps \Cref{eq:first} and \Cref{eq:second}, so without loss of generality we may assume that $w \cdot x \geq \tau$ for at least $p|T|/2$ points in $T$; under this assumption we will show that \Cref{eq:first} holds. 

Let $\tbg$ denote the result of conditioning $\bg$ on being $\alpha$-lucky for $w$, and note that by \Cref{lem:gauss-lb}, $\bg$ is $\alpha$-lucky with probability $q_1 := 2^{-O(\sqrt{n}\ln(n))}.$  We have that
\begin{equation} \label{eq:humbug}
{\frac{\Ex_{\tbg}\sbra{\abs{\cbra{x  \in T \cap R_+(\tbg): w \cdot x \geq \tau}}}}{\Ex_{\tbg}\sbra{\abs{\cbra{\xref  \in T \cap R_+(\tbg): w \cdot \xref \leq 0}}}}}=
{\frac {\sum_{x \in T: w \cdot x \geq \tau}\Pr_{\tbg}[x \in R_+(\tbg)]}{\sum_{\xref \in T: w \cdot \xref \leq 0}\Pr_{\tbg}[\xref \in R_+(\tbg)]}
}.
\end{equation}
There are at least $p|T|/2$ summands in the numerator and at most $|T|$ summands in the denominator of the RHS of \Cref{eq:humbug}. By \Cref{lem:advantage-simplified} each numerator-summand is at least $n^{7/3}$ times each denominator-summand, so recalling that $p={\frac 1 {4n}}$, \Cref{eq:humbug} gives 
\begin{align*}
\overbrace{\Ex_{\tbg}\sbra{\abs{\cbra{x  \in T \cap R_+(\tbg): w \cdot x \geq \tau}}}}^{=A}
&\geq {\frac p 2}n^{7/3} 
{\Ex_{\tbg}\sbra{\abs{\cbra{\xref  \in T \cap R_+(\tbg): w \cdot \xref \leq 0}}}}\\
&= \overbrace{{\frac {n^{4/3}} 8}}^{=B}
\overbrace{{\Ex_{\tbg}\sbra{\abs{\cbra{\xref  \in T \cap R_+(\tbg): w \cdot \xref \leq 0}}}}}^{=C}.
\end{align*}

Additionally, note that for any  $x$ with $w \cdot x \geq \tau$, by \Cref{eq:parsnip-simple} and the upper bound $t_\alpha(x) \leq O(n^{1/4} \sqrt{\ln n} )$ established near the end of the proof of \Cref{lem:advantage-simplified}, we have
\[
\Prx_{\tilde{\bg}}[x \in R_{+}(\tbg)] \geq 2^{-O(\sqrt{n} \ln(n))},
\]
and hence by linearity of expectation we get that
\[\overbrace{\Ex_{\tbg} \left[ \left| \{x \in T \cap R_{+}(\tbg): w \cdot x \geq \tau \} \right| \right]}^{=A} \geq \overbrace{p |T| 2^{-O(\sqrt{n} \ln(n))}}^{=D}.
\]

Using $2A/B - C \geq A/B \geq  D/B$ and linearity of expectation, we get that
    \begin{equation} \Ex_{\tbg} \left [ \overbrace{\frac{16}{n^{4/3}} \cdot \left| \{x \in T \cap R_{+}(\tbg): w \cdot x \geq \tau \} \right|}^{=2A/B} - \overbrace{\left| \{\xref \in T \cap R_{+}(\tbg): w \cdot \xref \leq 0\} \right|}^{=C} \right ] \geq \overbrace{ 2^{-O(\sqrt{n} \ln(n))}|T|}^{=D/B}.
    \label{eq:cheesepuff}
    \end{equation}
Since the quantity inside the expectation on the LHS is a real random variable $\bz$ that always takes values at most
$\frac{16}{n^{4/3}}|T| =: u$, by \Cref{lem:reverse-Markov} (``reverse Markov'') we have that $\Pr[\bz \geq {\frac \kappa 2} u] \geq {\frac \kappa 2}$ where $\kappa = {\frac {n^{4/3}}{16}}2^{-O(\sqrt{n} \ln(n))}$; in other words, with probability at least ${\frac {n^{4/3}}{32}}2^{-O(\sqrt{n} \ln(n))}=2^{-O(\sqrt{n} \ln(n))}=:q_2$, we have
\[
\frac{16}{n^{4/3}} \cdot \left| \{x \in T \cap R_{+}(\tbg): w \cdot x \geq \tau \} \right| 
- \left| \{\xref \in T \cap R_{+}(\tbg): w \cdot \xref \leq 0\} \right| \geq 2^{-O(\sqrt{n} \ln(n))}|T| .
\]
Note that this implies both 
\[
\left| \{\xref \in T \cap R_{+}(\tbg): w \cdot \xref \leq 0\} \right| \leq 
\frac{16}{n^{4/3}} \cdot \left| \{x \in T \cap R_{+}(\tbg): w \cdot x \geq \tau \} \right| 
\leq
\frac{16}{n^{4/3}} \cdot \left| T \cap R_{+}(\tbg) \right| 
\]
and
\[
\left|T \cap R_+(\tbg)\right| 
\geq  2^{-O(\sqrt{n} \ln (n))} |T|.
\]
This means that the probability in \Cref{eq:first} is at least $q_1 q_2 =  2^{-O(\sqrt{n} \ln (n))}$, recalling the definitions of $q_1$ and $q_2$ from earlier in the proof, and the proof is complete.
\end{proof}

\subsubsection{Achieving Non-Trivial Accuracy on the Sample $S$}

We now have the ingredients to show that with high probability, \warmup~constructs a ``simple'' hypothesis $h$ that correctly classifies significantly more than half of the points in the input data set $S$:

\begin{lemma} [Achieving non-trivial accuracy on a fixed sample] \label{lem:weak-learn-over-sample}
Suppose that the input data set $S$ for $\warmup$ is a sequence of $|S| = 2^{\Omega(\sqrt{n}\log n)}$ examples that are labeled according to some 
intersection of two halfspaces $f = \sgn(w^{(1)} \cdot x) \wedge \sgn(w^{(2)} \cdot x)$.  Then with probability at least $19/20$, $\warmup$ outputs a hypothesis $h: \R^n \to \bits$ that correctly classifies at least
${\frac 1 2} + \gamma$  fraction of the examples in $S$, where $\gamma := 2^{-O(\sqrt{n} \ln(n))}$.
Moreover, the hypothesis class ${\cal H}$ of all hypotheses that can be generated by \warmup\ has VC dimension at most $O(n^2).$
\end{lemma}

\begin{proof}
Our goal is to establish the following:

\begin{claim} \label{claim:goal}
For 
 each execution of the 2(a)-2(e) loop, 
 there is at least a $2^{-\tilde{O}(\sqrt{n})}$ probability that either Lines~2(a)-2(c) or Lines~2(d)-2(e) succeed in constructing a hypothesis $h: \R^n \to \bits$ that correctly classifies at least ${\frac 1 2} + 2^{-O(\sqrt{n}\ln(n))}$ fraction of points in $S$.
 \end{claim}
 
This is because given  \Cref{claim:goal}, $\Pr[$all $2^{\tilde{O}(\sqrt{n})}$ repetitions of the loop fail to construct such a hypothesis$]$  is at most $(1 - 2^{-\tilde{O}(\sqrt{n})})^{2^{\tilde{O}(\sqrt{n})}} \ll 1/10$, as desired. 

\begin{proofof}{\Cref{claim:goal}}
By \Cref{thm:forster}, the set $S'$ obtained in Line~1 is in $(3/2)$-radial isotropic position with respect to $V$, and hence by \Cref{lem:rad-iso-large-margin} $S'$ has a $(p = {\frac 1 {4n}},\tau={\frac 1 {2 \sqrt{n}}})$-margin with respect to the vector $w' = \frac{AP_V w^{(1)}}{\|AP_V w^{(1)}\|}.$
Consider a particular execution of the 2(a)-2(e) loop. By \Cref{lem:filtering-simple}, either (I) the event whose probability is lower bounded in \Cref{eq:first}, or (II) the event whose probability is lower bounded in \Cref{eq:second}, occurs with probability at least $2^{-O(\sqrt{n \ln(n)})}$ (taking the set $T$ of \Cref{lem:filtering-simple} to be $S'$ in both cases).

We suppose first that event (II) occurs.  In this case, by definition of event (II), at most a ${\frac {16}{n^{4/3}}}=o_n(1)$ fraction of the points in $S$ that are handled by the third line in the definition of $h_{b_1,b_2,-}$ are misclassified by $h_{b_1,b_2,-}$.  Hence for a suitable choice\footnote{To be specific, $b_1$ is the majority label of the examples in $S \setminus V$ and $b_2$ is the majority label of the examples $x \in S \cap V$ such that $\frac{Ax}{\|Ax\|} \not \in R_-(\bg)$.} of the bits $b_1,b_2 \in \bits$,
by the definition of event (II) the hypothesis $\bh_{b_1,b_2,-}$ is correct on at least a ${\frac 1 2} + 2^{-O(\sqrt{n} \ln(n))}$ fraction of points in $S$, as desired.

Next, we suppose that event (I) occurs.  In this case, by definition of event (I) and a union bound, a sample of $n \ln(n)$ points drawn uniformly and independently from $R_+(\bg) \cap S'$ as in Line~3(e)(i) has probability at most $o_n(1)$ of \emph{not} being labeled according to the halfspace $\sign(A^{-1} P_V w^{(2)} \cdot x)$
in $\R^n$, so with probability at least $1-o_n(1)$ the algorithm reaches Line~2(e)(iii). Moreover, since the class of halfspaces over $\R^n$ has VC dimension at most $n+1$, standard uniform convergence results (see e.g.~part~(1) of Theorem~6.8 of \cite{SSBD14}) give that with probability $1-o_n(1)$, any halfspace consistent with a random sample of $n \ln(n)$ points drawn uniformly from $R_+(\bg) \cap S'$ has accuracy $1-o_n(1)$ on $R_+(\bg) \cap S'$.
Hence for a suitable choice of the bits $b_1,b_2 \in \bits$,
by the definition of event (I) the hypothesis $\bh_{b_1,b_2,+}$ is correct on at least a ${\frac 1 2} + 2^{-O(\sqrt{n} \ln(n))}$ fraction of points in $S$, as desired.
This concludes the proof of \Cref{claim:goal}.
\end{proofof}

To finish the proof of \Cref{lem:weak-learn-over-sample}, we observe that any hypothesis defined in Lines~2(c) or 2(e)(iii) must belong to the class ${\cal H}$ of all functions $\R^n \to \bits$ that are of the form ``if $x \in V$ then output $b_1$; otherwise, if ${\frac {Ax}{\|Ax\|}} \in h'$ then output $b_2$; otherwise output $h''$'' where $b_1,b_2$ are fixed bits and $h',h''$ are halfspaces (note that $h''$ is the constant $-1$ in Line~2(c)).  For $h' = \sign(v \cdot x -\theta),$ the indicator of ${\frac {Ax}{\|Ax\|}} \in h'$ can be expressed as the intersection of a degree-2 polynomial threshold function and another halfspace, more precisely, as $\sgn\left( (A^T v \cdot x)^2 - \theta^2 \|Ax\|^2 \right) \land \sgn\left(  A^T v \cdot x \right)$. Thus, we can write any function in ${\cal H}$ as $g(f_1,f_2,f_3,f_4)$, where $g: \bits^4 \to \bits$ is a four-variable Boolean function; $f_1$ is the indicator function of a linear subspace $V$ of $\R^n$; $f_2$ is a degree-2 polynomial threshold function; and $f_3,f_4$ are both halfspaces over $\R^n$.
The VC dimension of the class of all linear subspaces of $\R^n$ is at most $n$ (see \cite{VCdim-subspace} for an easy argument establishing this); the VC dimension of the class of all degree-two polynomial threshold functions is well known to be $O(n^2)$; the VC dimension of the class of all halfspaces over $\R^n$ is well known to be $n+1$; and the VC dimension of the class of all functions $g: \bits^n \to \bits$ that  depend only on the first four coordinates is clearly $O(1)$. Given this, it follows directly from standard arguments (see the proof of Theorem~3.6 of \cite{KearnsVazirani:94}) that the VC dimension of the hypothesis class ${\cal H}$ is $O(n^2)$, and \Cref{lem:weak-learn-over-sample} is proved.
\end{proof}

By \Cref{lem:weak-learn-over-sample}, \warmup~is a $2^{\tilde{O}(\sqrt{n})}$-time algorithm that achieves accuracy at least $1/2 + 2^{-O(\sqrt{n} \log(n))}$ on \emph{any} sufficiently large fixed input sample that is labeled according to an intersection of two halfspaces, and  moreover it does this using a ``simple'' hypothesis (belonging to a hypothesis class of low VC dimension)  Given this, it is straightforward to prove \Cref{thm:2} using standard machinery from learning theory:  a standard generalization error argument shows that \Cref{lem:weak-learn-over-sample} easily yields a weak learning algorithm that achieves non-trivial advantage over random guessing on any input distribution, and then standard hypothesis boosting algorithms such as \cite{Schapire:90,Freund:95} give a full-fledged PAC learning algorithm. We provide full details, and complete the proof of \Cref{thm:2}, in \Cref{ap:proof-theorem-2}.


\section{Learning Functions of Halfspaces}
In this section we prove our main result:

\begin{theorem}
[Restatement of \Cref{thm:mainintro}]
\label{thm:main}
For any $k$,
there is a distribution-free PAC learning algorithm for the class of functions of $k$ halfspaces over $\mathbb{R}^n$, which learns to accuracy $\eps$ and confidence $1 - \delta$, running in time $\poly(2^{\sqrt{n} \log^{O(k)}(n)}, 1/\eps, \log(1/\delta))$, and using $\poly(2^{\sqrt{n} \log^{O(k)}(n)}, 1/\eps, \log(1/\delta))$ examples.
\end{theorem}

\begin{remark} \label{rem:large-k}
There is a simple ``brute-force'' algorithm that PAC learns any function of $k$ halfspaces in time $\poly(2^{2^k + O(nk^2)}, 1/\eps,\log(1/\delta))$ (see \Cref{ap:brute-force} for the straightforward argument that establishes this).  If $k \geq {\frac {c \log n}{\log \log n}}$ for any  constant $c>0$, then $2^{2^k + O(nk^2)} \leq 2^{\sqrt{n} \log^{O(k)}(n)}$, and hence \Cref{thm:main} is immediate for such large values of $k$. The most interesting values of $k$ for us are small values of $k$ satisfying $k \leq {\frac {c \log n}{\log \log n}}$, and in the rest of this section, without loss of generality, we will suppose that $k \leq {\frac {c \log n}{\log \log n}}.$
\end{remark}

We further remark that we have made no attempt to optimize the hidden constant in the big-Oh notation of \Cref{thm:main}.

\medskip

\noindent {\bf High-level overview.}  Recall that at a  high level, our algorithm for learning an intersection of two halfspaces is based on the fact that we can find a region $R \subseteq \R^n$ which (i) contains ``not too few'' of the Forsterized sample points $S'$, and (ii) is such that one of the halfspaces is almost constant over the points in $S' \cap R$. To learn a function of $k$ halfspaces, we will design an algorithm that similarly attempts to find a region, containing ``not too few'' of the sample points $S'$, where \emph{all $k$ halfspaces} are almost constant and thus so is the function.

In order to do this we will need some new ideas. To get started, it will be helpful to understand what goes wrong with the previous approach if we try to use it to learn an intersection of \emph{three} halfspaces given by vectors $w^{(1)}, w^{(2)}$ and $w^{(3)}$. After making a good guess $\bg^{(1)}$ for the first halfspace and restricting to (say) $R_+(\bg^{(1)})$,  at most an inverse polynomial fraction of points in $S' \cap R_+(\bg^{(1)})$ satisfy $w^{(1)} \cdot x \leq 0$. After Forsterizing the points in $S' \cap R_+(\bg^{(1)})$, we can then make a good guess $\bg^{(2)}$ for the second halfspace and intersect with the region (say) $R_+(\bg^{(2)})$; after doing this, at most an inverse polynomial fraction of points in $S' \cap R_+(\bg^{(1)}) \cap R_+(\bg^{(2)})$ satisfy  $w^{(2)} \cdot x \leq 0$. However, since $S' \cap R_+(\bg^{(1)}) \cap R_+(\bg^{(2)})$ may be only a very tiny (much less than inverse polynomial) fraction of $S' \cap R_+(\bg^{(1)})$,
it is entirely possible that we no longer have that a small fraction of points in $S' \cap R_+(\bg^{(1)}) \cap R_+(\bg^{(2)})$ satisfy  $w^{(1)} \cdot x \leq 0$.

To address this problem, we show that if many points in $S' \cap R_+(\bg^{(1)}) \cap R_+(\bg^{(2)})$ satisfy $w^{(1)} \cdot x \leq 0$ after intersecting with $R_+(\bg^{(2)})$, then it turns out that it must be the case that there was a ``not too small'' fraction of points in $S' \cap R_+(\bg^{(1)})$
that had an ``unexpectedly large'' margin with respect to $w^{(1)}$, specifically a margin  $\gg \frac{1}{\sqrt{n}}$.
We can then take advantage of this fact to better filter points, i.e.,~to construct a region that is ``more pure'' for the halfspace $w^{(1)}.$ Intuitively, this then makes progress towards learning, as if we ever find a region where a halfspace is constant on all but a $2^{-\sqrt{n} \log (n)}$ fraction of points, then we could produce a weak learner by learning a function of two halfspaces over the points in this region. 

At a high level, our algorithm is based on repeatedly applying this simple idea. However, quite a bit of care and technical work is required to set up the algorithm and its analysis so that everything works out.

\subsection{The Algorithm}
\label{sec:algorithm}
We now describe the algorithm for constructing a weak hypothesis with non-trivial accuracy over a fixed sample of examples labeled according to any function of $k$ halfspaces.
Similar to  \Cref{sec:warmup}, such an algorithm easily yields a weak PAC learning algorithm for any distribution ${\cal D}$, which in turn yields a strong PAC learning algorithm via standard boosting techniques (see \Cref{sec:strong-PAC-any-of-k}).

To aid the reader in digesting the algorithm and its subsequent analysis, we adhere to the following conventions:  variables $t,t'$, etc.~denote a ``time step'' (an execution of the Line~\ref{line:ltf-func-inner-loop} loop), and variables $i,j,\ell$, etc.~denote an index in $[k]$.

\begin{algorithm}[!hbtp]
\addtolength\linewidth{-2em}
\vspace{0.5em}

\textbf{Input:} A sequence $S$ of $2^{\sqrt{n} \log^{\Omega(k)}(n)}$ labeled examples $(x,y)$ where each $x \in \R^n$ and each $y = f(x)$ for an unknown $f: \R^n \to \bits$ that is an arbitrary function of $k$ halfspaces, $f(x)=g(\sign(w^{(1)} \cdot x),\dots,\sign(w^{(k)} \cdot x)).$ \\[0.25em]
\textbf{Output:} A hypothesis $h: \R^n \to \bits$ or FAIL.\\[0.5em]

\learnfunc($S$):\\[0.25em]

\begin{enumerate}
        \item For $i = 1, \dots k-1, k$, set $\beta_i := \frac{\log^{5(k-i+1)}(n)}{n^{1/4}}$. 
        \label{line:beta}

        \item Repeat $2^{\sqrt{n} \log^{O(k)}(n)}$ times:
        \label{line:outer-loop}
        \begin{enumerate}
        \item Let $S_0 := S$ and let $V_1, A_1, \bS_1 \gets \forster(S_0)$.
                   
        \item Set $\bu_1= \cdots=  \bu_k = 0 $.
        \item For $t = 1, \dots, \log^k(n)$:
        \label{line:ltf-func-inner-loop}
        \begin{enumerate}
            \item If $\bu_i \not = 0$ for all $i \in [k]$:
            \begin{enumerate}
                \item For $b_1, b_2, b_3 \in \{\pm 1\}$, let
                \[\bh_{b_1,b_2,b_3}(x) := \begin{cases} b_1 & \text{if~} x \not \in V_t \\
                b_2 & \text{~if~} x \in V_t \text{~and~}\exists t' \in [t-1] \text{~s.t.~} \frac{A_{t'}x}{\|A_{t'}x\|} \not \in R_{\bs_{t'}}^{\beta_{\br_{t'}}}(\bg^{(t')}) \\ 
                b_3 & \text{~if~} x \in V_t \text{~and~}\forall t' \in [t-1],  \frac{A_{t'}x}{\|A_{t'}x\|} \in R_{\bs_{t'}}^{\beta_{\br_{t'}}}(\bg^{(t')})\end{cases}.\]
                \item If some $\bh_{b_1, b_2, b_3}$, agrees with at least $1/2 + 2^{-\sqrt{n} \log^{O(k)}(n)}$ of the examples in $S$, output that $\bh_{b_1,b_2,b_3}$ and halt. 
                \item Otherwise, exit the inner loop on Line \ref{line:ltf-func-inner-loop} and return to the next execution of the outer Line~\ref{line:outer-loop} loop.
            \end{enumerate}

            \item Sample a random $\br_t \in  [k]$. \label{line:sample-r}
            \item Sample $\bg^{(t)} \sim \calN \left(0,\frac{1}{n} I_n \right)$.
            \label{line:sample-g}
            \item Set $\bu_{\br_{t}} = t$ and $\bu_i = 0$ for all $i > \br_{t}$.
            \label{line:u-update}
            \item Randomly choose $\bs_t \sim \{-, +\}$ and let $R_{\bs_t}^{\beta_{\br_t}}(\bg^{(t)}) := \{x \in \mathbb{R}^{n}: \bg^{(t)} \cdot x \geq \bs_t \cdot \beta_{\br_t} \}$.
            \label{line:sample-s}
            \item Set $V_{t+1},A',\bS_{t+1} \gets \forster(\bS_{t} \cap R_{\bs_t}^{\beta_{\br_t}}(\bg^{(t)}))$ and set $A_{t+1} \gets A' A_t$.         \label{line:forsterize-again}
        \end{enumerate}
        \end{enumerate}
    \item Return FAIL
\end{enumerate}

\caption{A Weak Learner for Functions of $k$ Halfspaces Over a Sample}
\label{alg:ltf-func}
\end{algorithm}

\subsubsection{High-level intuition and explanation of \Cref{alg:ltf-func}}

At the highest level, similar to the warm-up algorithm, for each halfspace $w^{(i)}$ we will attempt to make  a ``lucky'' (see \Cref{def:lucky}) guess $\bg^{(t)}$  and restrict to some region $R_{-}^{\beta_j}(\bg^{(t)})$ or $R_{+}^{\beta_j}(\bg^{(t)})$.
The  hope is to in this way construct a region $R \subseteq \R^n$ such that, writing $S'$ for the original data set $S$ after the sequence of linear transformations from the various Forsterizations have been performed on it, the set $S' \cap R$ has ``not too few'' points and is such that for each $i \in [k]$, either almost all of the points in $S' \cap R$ satisfy $\sgn(w^{(i)} \cdot x) \leq 0$ or almost all of them satisfy $\sgn(w^{(i)} \cdot x) \geq 0$. The outermost Line~\ref{line:outer-loop} loop goes over many attempts to have this happen.
Almost all of our discussion and analysis focuses on an execution of the Line~\ref{line:outer-loop} loop in which all of the guesses for the $\bg^{(t)}$'s (as well as some other guesses that we will discuss later) are ``lucky.''

In order for our analysis to establish that such a guessing-based procedure works with sufficiently high probability (at least $2^{-\sqrt{n}\log^{O(k)}(n)}$), in each execution of the Line~\ref{line:outer-loop} loop our algorithm must make a sequence of $\log^k(n)$ many guesses $\bg^{(1)},\bg^{(2)},\dots$ that may involve multiple guesses for each of the $k$ halfspaces in the target; this corresponds to the inner Line~\ref{line:ltf-func-inner-loop} loop.   The $\bu_1,\dots,\bu_k$ variables in our algorithm keep track of the progress that has been made across this sequence of guesses.

\medskip
\noindent {\bf The leaderboard.}
To explain how the $\bu_j$'s keep track of this in more detail, recall that for $i \in [k]$ we view a halfspace $w^{(i)}$ as being ``fixed'' on a region $R \subseteq \R^n$ if it takes the same value (either $+1$ or $-1$) on almost all of the transformed input examples that lie in $R$. 
Our analysis will employ a measure of ``how effectively'' each halfspace in $w^{(1)},\dots,w^{(k)}$ has been fixed so far. (We will elaborate much more on what it means for a halfspace $w^{(i)}$ to be ``fixed effectively'' later when we discuss the notion of ``quality'' (see \Cref{def:quality}); for now, we remark that ``quality''  involves both the numerical margin achieved on a suitable subset of examples as well as how large that subset is.)

In particular, the reader should have in mind a ``leaderboard'' of how effectively the various halfspaces have been fixed, with the halfspace $w^{(i_1)}$ that has been fixed most effectively in position 1 of the leaderboard, the halfspace $w^{(i_2)}$ that has been fixed second most effectively in position 2, and so on. 
We stress that the algorithm has no access at all to this leaderboard; it should be thought of as a helpful tool for analyzing fortuitous executions of the Line~\ref{line:outer-loop} loop in which all guesses are ``lucky'' as alluded to above.
The value of $\bu_j$ indicates the time step at which the halfspace in position $j$ of our leaderboard (namely, the halfspace that has been fixed ``the $j$-th most effectively'') was last fixed, i.e.~entered its current position on the leaderboard. 
In more detail, 

\begin{quote}
(*) 
The value of $\bu_j$ at the end of Line~\ref{line:u-update}
of any time step $t$ in the execution of the inner loop, which we denote $\bu_j(t)$ throughout our analysis, is the most recent time step $t' \leq t$ at which a $\bg^{(t')}$ was guessed that resulted in some target halfspace, which we denote $w^{(\ind(\bu_j(t)))}$,
being fixed 
``the $j$-th most effectively'' among all the halfspaces in $w^{(1)},\dots,w^{(k)}$ that are currently fixed.
If $\ell<k$ of the target halfspaces $w^{(\cdot)}$ are  fixed at time $t$, then $\bu_{\ell+1}(t)=\cdots=\bu_k(t)=0.$ 
\end{quote}
\noindent

It follows from (*) above that the value of $\ind(\bu_j(t))$ is the index in $[k]$ of which halfspace was fixed at time $\bu_j(t)$, and is also the index of the halfspace in position $j$ of the leaderboard at time $t$.  

We now explain in more detail the objectives of the guesses $\bg^{(t)},\br_{t}$ and $\bs_t$ that our algorithm makes in each execution of the inner Line~\ref{line:ltf-func-inner-loop} loop, starting with the guesses $\bg^{(t)}.$
Consider the start of time step $t$ of the inner loop, and suppose that at this time step $\bu_1(t),\dots,\bu_a(t) > 0$ and $\bu_{a+1}(t)=\cdots=\bu_k(t)=0$, meaning that halfspaces $w^{(\ind(\bu_1(t)))},\dots,w^{(\ind(\bu_a(t)))}$ are currently fixed by the guesses $\bg^{(1)},\dots,\bg^{(t-1)}$ that have been made thus far,
and are in positions $1,\dots,a$ of the leaderboard. Intuitively, at time step~$t$ the algorithm seeks to make a guess $\bg^{(t)}$ that \emph{causes a new target halfspace that is not currently fixed} to become fixed and hence take some position $a' \leq a+1$ on the leaderboard.  The position $a'$ that it takes will be $a+1$ if it is fixed less effectively than any of the currently-fixed halfspaces. Importantly, if $a'$ is less than $a+1$ (which is entirely possible if the newly-fixed halfspace is fixed more effectively than some halfspace currently on the leaderboard), then the halfspaces $w^{(\ind(\bu_{a'+1}))},\dots,w^{(\ind(\bu_{a}))}$ are ``unfixed'' and removed from the leaderboard, and consequently the values of $\bu_{a'+1},\dots,\bu_a$ are all updated to zero (see Line~\ref{line:u-update}). 
No matter what value $a'$ is in $[a+1]$, we set $\bu_{a'}$ to the current timestep $t$ (again see Line~\ref{line:u-update}), and the value of $\ind(t)=\ind(\bu_{a'})$ in our analysis becomes the index of the newly fixed halfspace.

We turn to explaining the objectives of the guesses $\br_t$ and $\bs_t$.  For $\br_t$,  we hope to set $\br_t$ to the value $a'$ corresponding to the position on the leaderboard of the newly-fixed halfspace (see Line~\ref{line:sample-r}).
Finally, Line~\ref{line:sample-s} attempts to guess a value for $\bs_t \in \{-,+\}$ corresponding to ``the right side of the halfspace $\bg^{(t)}$'' where there are ``many points with a large margin'' in the sense of \Cref{def:margin}.

Finally, we remark that the purpose of Line~\ref{line:forsterize-again} is to perform another round of Forsterization on the relevant set of points so that we can proceed to the next iteration of the Line~2(d) loop.

\subsubsection{Overview of ingredients of the algorithm and its analysis} \label{sec:overview}

Since the algorithm and its analysis are somewhat intricate, for the reader's convenience we now give a more detailed item-by-item overview of some of the key objects in the algorithm and its analysis. Our discussions of $\bu_j(t)$ and $\ind(t)$ below reiterate points that were made earlier, but we hope that this may help the reader form a clearer mental picture of what is going on.  Our discussions of $w^{(i)}(t)$, $\qual_j(t)$, $\imp_j(t)$, and $\mimp_j(t)$ below introduce new conceptual objects and ideas that play an important role in our analysis. The reader is encouraged to refer back to this subsection when reading the later detailed sections of our technical analysis.  

A condensed version of this overview, mentioning also some additional objects in the algorithm and its analysis, is provided  in \Cref{fig:figure2}.
In the following overview and in \Cref{fig:figure2}, we note that whenever an object has ``$(t)$'' appended to it, this refers to the object at the end
of time step $t$ of the Line~\ref{line:ltf-func-inner-loop} loop; thus, for example, $\bu_j(t)$ is the value of $\bu_j$ when Line~\ref{line:forsterize-again} has been executed
for the $t$-th time.
To aid the reader in distinguishing between objects which occur in the algorithm and objects which occur only in the analysis, we use \violet{violet font} both in \Cref{fig:figure2} and in the item-by-item overview for objects which  occur only in the analysis of the algorithm and not in the algorithm itself.

We stress again that our discussion of the quantities below should be thought of as only referring to an execution of the Line~1(d) loop in which ``all guesses are lucky'' (see the definition of a ``good total execution'', \Cref{def:good}).

\begin{itemize}

    \item The meaning of $\bu_j(t)$ was explained earlier in (*); to recall, $\bu_j(t)$ is the time step $t' \leq t$ at which the halfspace currently (at the end of time step $t$) in position $j$ of our leaderboard (namely, the halfspace that has been fixed ``the $j$-th most effectively'') entered that position. We write $\vec{\bu}(t)$ to denote the vector $(\bu_1(t),\dots,\bu_k(t))$ and we write $\nnz(\vec{\bu}(t))$ to denote the number of nonzero entries in $\vec{\bu}(t)$, i.e.~the number of halfspaces that have currently been fixed at time $t$. Note that $\nnz(\vec{\bu}(t))$ may go up and down as $t$ increases, as halfspaces are fixed and unfixed as described earlier.  

    \item \violet{As mentioned earlier, $\ind(t)$ takes as input a time step $t$ and outputs the \underline{ind}ex $i \in [k]$ of which target halfspace $w^{(i)}$ is fixed by the guess $\bg^{(t)}$ at time $t$. 
    Therefore, the index of the halfspace at position $j$ of our leaderboard at time $t$ is $\ind(\bu_j(t))$, for $j \in \{1,\dots,\nnz(\vec{\bu}(t))\}.$}
    
    \item $V_t$ is the subspace of $\R^n$ that the algorithm has ``zoomed in on'' after the first $t$ rounds of Forsterization, and  $A_t: V_t \to V_t$ is the linear transformation obtained by composing all $t$ of the Forsterization linear transformations that have been performed thus far (see Line~\ref{line:forsterize-again}).
    
    \item The value $\br_t \in [k]$ is the ``\underline{r}ank'' on the leaderboard of the halfspace $w^{(\ind(\bu_{\br_t}))}$ that is fixed at time step $t$, i.e.,~which position $\br_t$ of the leaderboard has $\bu_{\br_t}$ updated at time step $t$.

    \item The set $\bS_{t+1}$ is the  (transformed-by-$A_t$-version-of) the subset of the original input set $S$ 
    that is still ``in play'' after time step $t$ (see Line~\ref{line:forsterize-again}), i.e.~the set of ``filtered and Forsterized'' points.  $\bS_{t+1}$ contains the points over which the algorithm will continue to attempt to fix all $k$ of the target halfspaces.

    \item 
    The set $R_{\bs_t}^{\beta_{\br_t}}(\bg^{(t)})= \{x \in \mathbb{R}^{n}: \bg^{(t)} \cdot x \geq \bs_t \cdot \beta_{\br_t} \}$ should be thought of as the \underline{r}egion of $\R^n$ consisting of those points $x$ that have a ``large margin vis-a-vis the $t$-th guessed halfspace $\bg^{(t)}$ on the side corresponding to $\bs_t$.''
    The points in $\bS_t \cap R_{\bs_t}^{\beta_{\br_t}}(\bg^{(t)})$ are the ones that are used as input to the $(t+1)$-st Forsterization step (see Line~\ref{line:forsterize-again}). In a good execution, $w^{(\ind(t))}$ will be fixed in this region.

    \item \violet{The vector $w^{(i)}(t)$ stands for $\frac{A_t^{-1}P_{V_t} w^{(i)}}{\|A_t^{-1}P_{V_t} w^{(i)}\|}$. In words, this is a vector in $V_t$ (normalized to be a unit vector) corresponding to the original vector $w^{(i)}$ after it has been projected to the currently-relevant subspace $V_t$ and transformed by the Forsterizations.  (Recall from \Cref{remark:Forster} that this is always a nonzero vector and hence $w^{(i)}(t)$ is indeed well-defined.)}

    \end{itemize}

    On a first read through the paper, the reader may wish to skip the following items and come back to them in the course of reading \Cref{sec:good-execution,sec:great-execution}.
    \begin{itemize}
    
    \item \violet{
    The quantity $\qual_j(t)$ capures the ``quality of the $j$-th slot in the leaderboard at time $t$.''  It measures how effectively the
    $\ind(\bu_j(t))$-th
    halfspace $w^{(\ind(\bu_j(t)))}$ (the one that occupies the $j$-th slot in the leaderboard at time $t$) was fixed at the time step $\bu_j(t)$
    when it was fixed. Quantitatively, it takes into account both the fractional size of the relevant subset of examples that have a ``large margin'' as well as how large that margin is; a higher value indicates higher quality.  See \Cref{def:quality} for a detailed definition.}

    \item \violet{
    The quantity $\imp_j(t)$ measures the ``\underline{imp}urity'' of the halfspace in position $j$ of the leaderboard at time step $t$.  It is equal to the fraction of points in the relevant subset $\bS_{t+1}$ of examples that satisfy $\sgn(w^{(\ind(\bu_{j}(t)))} \cdot x) \neq \bs_{\bu_j(t)}$ after the $t$-th iteration of Line~\ref{line:forsterize-again}. See \Cref{def:finefilter} for a detailed definition.
    }

    \item \violet{The quantity  $\mimp_j(t)$ (standing for ``\underline{m}odified \underline{imp}urity'') is an upper bound on the impurity $\imp_j(t)$ and is a technical convenience for us. In particular, we would like to argue that the impurity increases by a small multiplicative factor at each step. Unfortunately, such a statement only holds if the impurity is not too small, leading us to use the modified quantity $\mimp_j(t)$ instead.}

\end{itemize}

We conclude this subsection by stating the main goal of \Cref{sec:good-execution} through \Cref{sec:filtering}. This is to prove the following crucial lemma:

\begin{lemma} [Non-tiny probability of generating a hypothesis with non-trivial accuracy over the sample]
    \label{lem:ltf-func-loop}
    With probability at least $2^{-\sqrt{n} \log^{O(k)}(n)}$ over the random choices of the $\bg^{(t)}$'s, the $\br_t$'s, and the $\bs_t$'s, an execution of the $\log^k(n)$ repetitions of the inner Line \ref{line:ltf-func-inner-loop} loop of \Cref{alg:ltf-func} outputs a hypothesis $h_{b_1,b_2,b_3}$ that correctly classifies ${\frac 1 2} + 2^{-\sqrt{n} \log^{O(k)}(n)}$ fraction of the examples in $S$.
\end{lemma}

We remark that since there are $\log^k(n)$ time steps $t$, and at each time step we guess ``the right value'' of $\br_t$ (respectively, $\bs_t$) with probability at least $1/k$ (respectively, at least $1/2)$, the reason that the success probability in \Cref{lem:ltf-func-loop} is only $2^{-\sqrt{n} \log^{O(k)}(n)}$ is because of the relatively low probability of successfully guessing vectors $\bg^{(1)},\dots,\bg^{(\log^k(n))}$ that satisfy our requirements.

As mentioned at the beginning of this subsection, similar to the warm-up, given \Cref{lem:ltf-func-loop} a standard analysis yields a PAC learning algorithm for any function of $k$ halfspaces; see \Cref{sec:strong-PAC-any-of-k} for details.

\begin{figure}[ht]
\begin{centering}
\begin{tabular}{|c|c|c|}
\hline
Symbol & Description & Reference \\
\hline
$\bu_j(t)$ & Time stamp of when the halfspace in position $j$ & \Cref{def:good}  \\
& of the leaderboard at time $t$ was last fixed & \\
\hline
\violet{$\ind(t)$} & \violet{Index in $[k]$ of which unfixed halfspace $w^{(\ind(t))}$} & \Cref{def:good}  \\
& \violet{gets fixed by guess $\bg^{(t)}$ at time $t$}  & \\
\hline
$V(t)$ & The subspace of $\R^n$ that the filtered & Line~\ref{line:forsterize-again}\\
& and Forsterized points lie in at time $t$ &  \\
\hline
$A(t)$ & Composition of the first $t$ Forsterization linear transformations; & Line~\ref{line:forsterize-again}\\
& ensures $\bS_{t+1}$ is in $3/2$-radial isotropic position w.r.t. $V(t+1)$& \\
\hline
$\bs_t$ & ``Side'' ($+$ or $-$) of the halfspace $\bg^{(t)} \cdot x \geq \bs_t \cdot \beta_{\br_t}$ to which & \Cref{def:good}  \\
& we are restricting at time $t$& \\ 
\hline
$\bg^{(t)}$ & The random vector, drawn from $\calN \left(0,\frac{1}{n} I_{n} \right)$, & Line~\Cref{line:sample-g} \\
& that is guessed at time $t$ & \\
\hline 
$\br_t$ & Index in $[k]$ specifying which position $\br_t$ of the leaderboard & \Cref{def:good}\\
& had $\bu_{\br_t}$ updated at time $t$ &  \\
\hline
$\bS_t$ & Set of filtered and Forsterized points that ``are  still in play''& Line~\ref{line:forsterize-again}\\
&  at time $t$; in $3/2$-radial isotropic position w.r.t.~$V(t)$ &  \\
\hline
$\beta_j$ &  $\frac{\log^{5 (k-j+1)}(n)}{n^{1/4}}$;  the threshold of the restricting region $\bg^{(t)}\cdot x \geq \bs_t \cdot \beta_j$ & Line~\ref{line:beta}\\
& attempting to fix the target halfspace in leaderboard position $j$
&  \\
\hline
$\alpha_j$ &  $10\beta_j$; a ``luckiness'' parameter for good guesses of the $\bg^{(t)}$'s & \Cref{sec:good-execution} \\
\hline
$R_{\bs_t}^{\beta_{\br_t}}(\bg^{(t)})$ & The new region $\{x \in \R^n: \bg^{(t)} \cdot x \geq \bs_t \cdot \beta_{\br_t}\}$ that we are & Line~\ref{line:sample-s}\\
& further restricting by
at time $t$ & \\
\hline 
\violet{$w^{(i)}(t)$} & \violet{ $\frac{A_t^{-1}P_{V_t} w^{(i)}}{\|A_t^{-1}P_{V_t} w^{(i)}\|}$; the transformed-and-projected version} & \Cref{remark:Forster}\\
& \violet{of the original target halfspace vector $w^{(i)}$ at time $t$} &  \\
\hline
\violet{$\qual_j(t)$} & \violet{The $j$-quality of the 
$\ind(\bu_j(t))$-th
halfspace when it was fixed,} & \Cref{def:quality} \\
& \violet{i.e.,~the quality of the $j$-th slot in the leaderboard at time $t$} &  \\
\hline
\violet{$\imp_j(t)$} & \violet{The fraction of points in $\bS_{t+1}$ 
satisfying $\sgn(w^{(\ind(\bu_{j}(t)))} \cdot x) \neq \bs_t$} & \Cref{def:finefilter} \\
& \violet{(on the ``wrong side'' of $w^{(\ind(\bu_j(t)))}$)} 
& \\
\hline
\violet{$\mimp_j(t)$} & \violet{Equals the maximum of $\imp_j(t)$ and $e^{-\sqrt{n} \beta_j^2 \log^{\qual_j(t)-1}(n)/3}$} & \Cref{def:finefilter} \\
\hline
\end{tabular}
\end{centering}
\caption{Description and Pointers for Notation For the Algorithm and Its Analysis.
}
\label{fig:figure2}
\end{figure}
\FloatBarrier

\begin{figure}
\begin{tikzpicture}
\tikzset{
  mystar/.style={shape=star,star points=5,minimum size=#1,star point ratio=2.1,rounded corners=#1/20}
}

\coordinate (p1) at (0,10.5);

\node[anchor=west, font=\crayonfont] at ($(p1) + (0,-0.2)$) {time $t = 65$:};

\draw[crayon draw, fill=violet!4] ($(p1) + (0, -3.1)$) rectangle ($(p1) + (16,-0.5)$);

\node[sticker star] at ($(p1) + (0.5, -1)$) {};
\node[sticker star] at ($(p1) + (15.4,-1)$) {};

\node[title text] at ($(p1) + (8, -1)$) {Leaderboard};

\begin{scope}[shift={(0,0)}]
  \coordinate (c1) at ($(p1) + (2, -2)$);
  \coordinate (c2) at ($(p1) + (6.2, -2)$);
  \coordinate (c3) at ($(p1) + (10.4, -2)$);

  \node[rank text] at (c1) {Rank 1:};
  \draw[crayon line] ($(c1)+( -1.6,-0.25)$) -- ($(c1)+(1.6,-0.25)$);
  \node[rank text] at ($(c1)+(0,-0.7)$) {$w^{(3)}$};

  \node[rank text] at (c2) {Rank 2:};
  \draw[crayon line] ($(c2)+(-1.6,-0.25)$) -- ($(c2)+(1.6,-0.25)$);
  \node[rank text] at ($(c2)+(0,-0.7)$) {$w^{(6)}$};

  \node[rank text] at (c3) {Rank 3:};
  \draw[crayon line] ($(c3)+(-1.6,-0.25)$) -- ($(c3)+(1.6,-0.25)$);
  \node[rank text] at ($(c3)+(0,-0.7)$) {$w^{(4)}$};
\end{scope}

\coordinate (t1) at ($(p1) + (1.9,6.7-10.5)$);
\def\rinc{4.2}
\def\cinc{-0.5}

\node[small math, anchor=center] at (t1) {$\qual_1\!\big(w^{(3)},65\big)=12$};
\node[small math, anchor=center] at ($(t1) + (\rinc,0)$) {$\qual_2\!\big(w^{(6)},65\big)=8$};
\node[small math, anchor=center] at ($(t1) + 2*(\rinc,0)$){$\qual_3\!\big(w^{(4)},65\big)=8$};

\node[small math, anchor=center] at ($(t1) + (0,\cinc)$) {\black{$u_1(65)=24$}};
\node[small math, anchor=center] at ($(t1) + (\rinc, \cinc)$) {\black{$u_2(65)=37$}};
\node[small math, anchor=center] at ($(t1) + (2*\rinc, \cinc)$) {\black{$u_3(65)=65$}};
\node[small math, anchor=center] at ($(t1) + (3*\rinc, \cinc)$) {\black{$u_4 = \dots = u_k = 0$}};

\node[small math, anchor=center] at ($(t1) + (0,2*\cinc)$) {$\ind(24)=3$};
\node[small math, anchor=center] at ($(t1) + (\rinc, 2*\cinc)$) {$\ind(37)=6$};
\node[small math, anchor=center] at ($(t1) + (2*\rinc, 2*\cinc)$) {$\ind(65)=4$};

\node[note, text width=15.6cm, align=left, anchor=west] at ($(t1) + (0.4,4.1) - (p1) - (1.9,6.7-10.5-0.4)$)
{A depiction of the leaderboard at time step $t=65$, along with some related quantities, during a good total execution of the line 2(c) loop. Information in \textcolor{violet!85!black}{violet} is used only in the analysis and is not available to the algorithm. At time step 65, we had $\br_{65} = 3$ and $\bg^{(65)}$ was an $\alpha_3$-lucky guess w.r.t.~$w^{(4)}$. At time step 24, $\bg^{(24)}$ was an $\alpha_1$-lucky guess w.r.t.~$w^{(3)}$, and at time step 37, $\bg^{(37)}$ was an $\alpha_2$-lucky guess w.r.t.~$w^{(6)}.$};

\coordinate (p2) at (0,2.8);
\node[anchor=west, font=\crayonfont] at ($(p2) - (0,0.4)$) {time $t = 66$:};

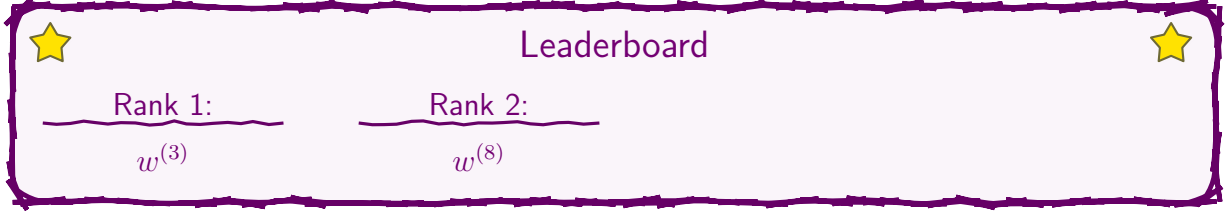
\draw[crayon draw, fill=violet!4] ($(p2) + (0, -0.5) - (0,2.8)$) rectangle ($(p2) + (16.0, 2.1) - (0,2.8)$);
\node[sticker star] at ($(p2) + (0.5, 1.6) - (0,2.8)$) {};
\node[sticker star] at ($(p2) + (15.4, 1.6) - (0,2.8)$) {};
\node[title text] at ($(p2) + (8,1.6) - (0,2.8)$) {Leaderboard};


\coordinate (d1) at ($(p2) + (2, -2)$);
\coordinate (d2) at ($(p2) + (6.2, -2)$);


\node[rank text] at (d1) {Rank 1:};
\draw[crayon line] ($(d1)+(-1.6,-0.25)$) -- ($(d1)+(1.6,-0.25)$);
\node[rank text] at ($(d1)+(0,-0.7)$) {$w^{(3)}$};

\node[rank text] at (d2) {Rank 2:};
\draw[crayon line] ($(d2)+(-1.6,-0.25)$) -- ($(d2)+(1.6,-0.25)$);
\node[rank text] at ($(d2)+(0,-0.7)$) {$w^{(8)}$};

\coordinate (t2) at ($(p2) + (1.9,6.7-10.5) - (0,0.2)$);
\def\rincb{4.2}
\def\cincb{-0.5}

\node[small math, anchor=center] at (t2) {$\qual_1\!\big(w^{(3)},66\big)=12$};
\node[small math, anchor=center] at ($(t2) + (\rinc,0)$) {$\qual_2\!\big(w^{(8)},66\big)=9$};

\node[small math, anchor=center] at ($(t2) + (0,\cinc)$) {\black{$u_1(66)=24$}};
\node[small math, anchor=center] at ($(t2) + (\rinc, \cinc)$) {\black{$u_2(66)=66$}};
\node[small math, anchor=center] at ($(t2) + (2*\rinc, \cinc)$) {\black{$u_3 = \dots = u_k = 0$}};

\node[small math, anchor=center] at ($(t2) + (0,2*\cinc)$) {$\ind(24)=3$};
\node[small math, anchor=center] at ($(t2) + (\rinc, 2*\cinc)$) {$\ind(66)=8$};

\node[note, text width=15.6cm, align=left, anchor=west] at ($(t2) + (0.4,4.1) - (p1) - (1.9,6.7-10.5-0.6)$)
{At time step 66, $\br_{66} = 2$, $\bg^{(66)}$ was an $\alpha_2$-lucky guess w.r.t.~$w^{(8)}$, and we have that  $\qual_2(w^{(8)},66) > \qual_2(w^{(6)},65)$.
Hence (cf.~item~(i) of \Cref{def:good}) $w^{(8)}$ replaces $w^{(6)}$ in the second slot of the leaderboard and all halfspaces with lower rank are reset.};
\coordinate (p3) at (0,-4.9);

\node[anchor=west, font=\crayonfont] at ($(p3) + (0,-0.2)$) {time $t = 67$:};

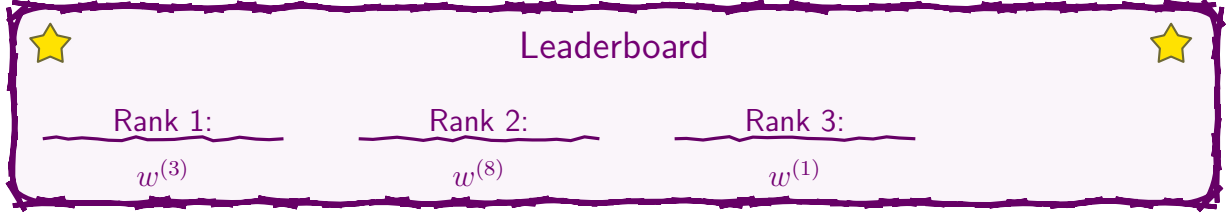
\draw[crayon draw, fill=violet!4] ($(p3) + (0, -3.1)$) rectangle ($(p3) + (16,-0.5)$);

\node[sticker star] at ($(p3) + (0.5, -1)$) {};
\node[sticker star] at ($(p3) + (15.4,-1)$) {};

\node[title text] at ($(p3) + (8, -1)$) {Leaderboard};

\begin{scope}[shift={(0,0)}]
  \coordinate (c1) at ($(p3) + (2, -2)$);
  \coordinate (c2) at ($(p3) + (6.2, -2)$);
  \coordinate (c3) at ($(p3) + (10.4, -2)$);

  \node[rank text] at (c1) {Rank 1:};
  \draw[crayon line] ($(c1)+( -1.6,-0.25)$) -- ($(c1)+(1.6,-0.25)$);
  \node[rank text] at ($(c1)+(0,-0.7)$) {$w^{(3)}$};

  \node[rank text] at (c2) {Rank 2:};
  \draw[crayon line] ($(c2)+(-1.6,-0.25)$) -- ($(c2)+(1.6,-0.25)$);
  \node[rank text] at ($(c2)+(0,-0.7)$) {$w^{(8)}$};

  \node[rank text] at (c3) {Rank 3:};
  \draw[crayon line] ($(c3)+(-1.6,-0.25)$) -- ($(c3)+(1.6,-0.25)$);
  \node[rank text] at ($(c3)+(0,-0.7)$) {$w^{(1)}$};
\end{scope}

\coordinate (t1) at ($(p3) + (1.9,6.7-10.5)$);

\node[small math, anchor=center] at (t1) {$\qual_1\!\big(w^{(3)},67\big)=12$};
\node[small math, anchor=center] at ($(t1) + (\rinc,0)$) {$\qual_2\!\big(w^{(8)},67\big)=9$};
\node[small math, anchor=center] at ($(t1) + 2*(\rinc,0)$){$\qual_3\!\big(w^{(1)},67\big)=4$};

\node[small math, anchor=center] at ($(t1) + (0,\cinc)$) {\black{$u_1(67)=24$}};
\node[small math, anchor=center] at ($(t1) + (\rinc, \cinc)$) {\black{$u_2(67)=66$}};
\node[small math, anchor=center] at ($(t1) + (2*\rinc, \cinc)$) {\black{$u_3(67)=67$}};
\node[small math, anchor=center] at ($(t1) + (3*\rinc, \cinc)$) {\black{$u_4 = \dots = u_k = 0$}};

\node[small math, anchor=center] at ($(t1) + (0,2*\cinc)$) {$\ind(24)=3$};
\node[small math, anchor=center] at ($(t1) + (\rinc, 2*\cinc)$) {$\ind(66)=8$};
\node[small math, anchor=center] at ($(t1) + (2*\rinc, 2*\cinc)$) {$\ind(67)=1$};

\node[note, text width=15.6cm, align=left, anchor=west] at ($(t1) + (0.4,4.1) - (p1) - (1.9,6.7-10.5-0.6)$)
{At time step~67 there is no unfixed halfspace that has higher $1$-quality than $w^{(3)}$ or higher $2$-quality than $w^{(8)}$.
At this time step we have $\br_{67}=3$, $\bg^{(67)}$ was an $\alpha$-lucky guess w.r.t.~$w^{(1)}$, and $\qual_3(w^{(1)},67)=4$.
Note that $w^{(1)}$ would have been placed into rank 3 on the leaderboard had its $3$-quality been any value from $0, 1, \dots 9$.};

\end{tikzpicture}
\caption{Example depictions of the leaderboard} \label{fig:leaderboard}
\end{figure}

\subsection{A Good Execution}
\label{sec:good-execution}

Throughout the execution of the algorithm, we make a number of guesses; in this section we explain in detail what are the desired outcomes of those guesses.

We start with the following basic definition:

\begin{definition}[$\alpha$-lucky] \label{def:lucky}
    We say that a vector $g$ is \emph{$\alpha$-lucky with respect to a vector $w \in \mathbb{S}^{n-1}$} if $w \cdot \bg \geq \alpha$.
\end{definition}

We will always set 
\[ \alpha_j:= 10 \beta_j\]
for all $j \in [k-1]$.

A crucial notion for us is the \emph{quality} of a halfspace:

\begin{definition}[$j$-Quality of a Halfspace]
\label{def:quality}
Given a 
a value $j \in [k-1]$,
we say that the 
\emph{$j$-quality of a target halfspace $w^{(i)}$ at time $t$},
denoted $\qual_j(w^{(i)},t)$, is the largest non-negative integer $q$ such that $\bS_{\bu_j(t)}$\footnote{For intuition, recall that $\bS_{\bu_j(t)}$ is the set of examples ``that were in play'' at the time $\bu_j(t)$ at which the halfspace $w^{(\ind(\bu_j(t))}$ that is in position $j$ on the leaderboard at time $t$ was last fixed, i.e.~entered that position on the leaderboard.} has a $(p, \tau)$-margin with respect to
$w^{(i)}$, 
where \[
\quad p = \frac{1}{4n} \exp(-n \beta_j^2 \tau/\log(n))
\quad \text{and} \quad
\tau = \frac{\log^{q}(n)}{2\sqrt{n}}.
\]
Most of the time we will be concerned with $\qual_j(w^{(\ind(\bu_j(t)))},t)$; to ease notation we will simply write $\qual_j(t)$ as shorthand for $\qual_j(w^{(\ind(\bu_j(t)))},t)$. 
We refer to this quantity $\qual_j(t)$ as the \emph{quality of the $j$-th slot in the leaderboard at time $t$.}

\end{definition}

Recall from \Cref{lem:rad-iso-large-margin} that since $\bS_{t'}$ is in $(3/2)$-radial isotropic position for all $t'$, the quantity  $\qual_j(w^{(i)},t)$ is indeed well-defined since $q=0$ satisfies the required conditions.
Note that since the margin parameter $\tau$ can never exceed 1, we will always have that $\qual_j(w^{(i)},t)$ is at most $O({\frac {\log n}{\log \log n}}).$
Also, note that since $\ind(t')$ is only defined for $t'>0$ (this should be clear from our discussion of $\ind(\cdot)$ already, and will become clearer in \Cref{def:good} below), the quantity $\qual_j(t)$ is only defined for $j \leq \nnz(\vec{\bu}(t)).$

Note the tension between the parameters $p$ and $\tau$ in \Cref{def:quality}: for larger $q$, the required margin $\tau$ is larger (which makes it easier to learn), but the relevant fraction of points $p$ with that margin is smaller (which makes it harder to learn). \Cref{def:quality} is carefully chosen to manage this tension in a way that will be useful for us;
intuitively, one should think of  halfspaces with higher quality as being algorithmically ``easier to handle'' via our techniques.

With \Cref{def:lucky} and \Cref{def:quality} in hand, we can now define a \emph{good time step} and a \emph{good total execution} of the loop on Line \ref{line:ltf-func-inner-loop}.
This definition also provides the formal definition of $\ind(t).$

\begin{definition}[Good time steps and good total executions.] \label{def:good}
    We say that a time step $t \in [\log^k(n)]$ of the loop on Line~\ref{line:ltf-func-inner-loop} is \emph{good} if $\bg^{(t)}$, $\br_t$ and $\bs_t$ satisfy the following conditions:

\begin{itemize}
    
    \item [$(i)$] Suppose $t$ is such that there exists at least one pair $(j,a) \in [k] \times [k]$, where $a \neq \ind(\bu_i(t))$ for any $i \in \nnz(\vec{\bu})$, such that
    $\qual_j(w^{(a)},t) > \qual_j(t)$.
    In words, $a$ is the index of an unfixed halfspace with higher $j$-quality than the halfspace corresponding to position $j$ on the current leaderboard.
    In this case, let $(j,a)$ denote the pair with smallest value of $j$ and, among those, smallest value of $a$; (I) $\br_t=j$, and (II) $\bg^{(t)}$ is $\alpha_{\br_t}(t)$-lucky for the halfspace $w^{(a)}$.  For the analysis, $\ind(t)$ is set to $a$.
    
    \item [$(ii)$] 
    The other possibility is that $t$ is such that no such $(j,a)$ pair exists as in $(i)$. In this case 
    (I$'$) $\br_t=\nnz(\vec{\bu}(t))+1$, and (II$'$) $\bg^{(t)}$ is $\alpha_{\br_t}(t)$-lucky for the target halfspace $w^{(a)}$ with the smallest value of $a \in [k]$ satisfying $a \neq \ind(\bu_i(t))$ for any $i \in \nnz(\vec{\bu}(t))$.
    For the analysis, $\ind(t)$ is set to that $a$.

    \item [$(iii)$] Finally, $\bs_t$  corresponds to the side of $w^{(\ind(t))}$ with ``the most high margin points certifying the $\br_t$-quality of $w^{(\ind(t))}$ over $\bS_{t}$'',
    i.e.~the side of $w^{(\ind(t))}$ with the most points in $\bS_t$ that have margin at least $\frac{\log^{\br_t}(n)}{2 \sqrt{n}}$ for $w^{(\ind(t))}$.

\end{itemize}
    
We say that a \emph{total execution} of all time steps of the loop on Line \ref{line:ltf-func-inner-loop} is \emph{good} if every time step $t \in [\log^k(n)]$ is good.
\end{definition}

We will use the above notation throughout this section. In particular, as described earlier, $\ind(t)$ denotes the halfspace for which $\bg^{(t)}$ is $\alpha_{\br_t}(t)$-lucky. 

In the rest of this subsection we record a few useful invariants that hold throughout any good total execution of the loop on Line \ref{line:ltf-func-inner-loop}. 

\begin{observation} \label{obs:quality}
 For any target halfspace $w^{(i)}$, in any good time step $t$ of the loop on Line \ref{line:ltf-func-inner-loop} we have that
        \[\qual_{j+1}(w^{(i)},t) \leq \qual_{j}(w^{(i)},t). \]
\end{observation}
\begin{proof}
\Cref{obs:quality} holds from the definition of of quality (\Cref{def:quality}) and the fact that the $\beta_j$'s form a decreasing sequence.
\end{proof}

 \begin{observation} \label{obs:increasing}
 In any 
 execution of the loop on Line \ref{line:ltf-func-inner-loop}, at any time step $t$ the sequence of values $\bu_1(t),\bu_2(t),\dots,\bu_k(t)$ is an increasing sequence followed by a sequence of $0$'s.
 \end{observation}
\begin{proof}
 This is a direct consequence of Line~\ref{line:u-update} of \Cref{alg:ltf-func}.
 \end{proof}

 \begin{observation}
     \label{obs:decreasing-quality}
 Consider any good total execution of the loop on Line \ref{line:ltf-func-inner-loop} and any time $t \in [\log^k(n)].$
 The sequence $\qual_1(t), \qual_2(t), \dots,\qual_{\nnz(\vec{\bu}(t))}(t)$ is non-increasing.
\end{observation}
\begin{proof}
\Cref{obs:decreasing-quality} is a consequence of \Cref{def:quality}, \Cref{obs:quality} and the fact that the $\beta_1,\beta_2,\dots$ values form a decreasing sequence.
\end{proof}

\subsection{Fine Filters}
\label{sec:great-execution}

We now turn to arguing that a good total execution of the inner loop on Line \ref{line:ltf-func-inner-loop} which ends at time step $t=T$ will find a region $R \subseteq \R^n$ such that (a) the set of points  $\bS_{T} \subset R$ is ``not too small,'' and (b) every target halfspace  is almost constant over the points in $\bS_T$. 
In particular, our goal is  to argue that after the last step $T$ of a good total execution of the loop, for the vast majority of points $x$ in $\bS_T$, we have that the corresponding point in $S$
satisfies $\sgn(w^{\ind(\bu_j(T))} \cdot x) = \bs_{\bu_j(T)}$
for all $j \in [k]$. Towards this goal, we make the following definition:

\begin{definition}[Fine Filter] \label{def:finefilter}
Given $t \in [\log^k(n)]$ and $j \in [k]$,
let $\imp_j(t)$ denote the fraction of points in $\bS_{t+1}$ satisfying $\sgn(w^{(\ind(\bu_{j}(t)))} \cdot x) \neq \bs_{\bu_j(t)}$ after Line~\ref{line:forsterize-again}, and let \[
\mimp_j(t) := \max \left(\imp_j(t), e^{-\sqrt{n} \beta_j^2 \log^{\qual_j(t) - 1}(n)/3} \right).
\]
We say that a guess $\bg^{(t)}$ is a 
\emph{fine filter} if the following hold (note that the outcome of $\bg^{(t)}$ determines $\bS_{t+1}$ and hence $\imp_j(t)$ and $\mimp_j(t)$):
\begin{itemize}
    \item [(1)] We have that 
     \[ \imp_{\br_t}(t) \leq e^{-\sqrt{n} \beta_{\br_t}^2 \log^{\qual_{\br_t}(t) - 1}(n)/3};\]
    \item [(2)]
    For all $j < \br_t$ we have that 
    \[\mimp_j(t) \leq \mimp_j(t-1) \cdot e^{O(\sqrt{n} \beta_{\br_t}^2 \log^{\qual_j(t)+1}(n))};  \quad \text{and}\]

    \item [(3)]
    \[
|\bS_{t+1}| \geq e^{-O(n \beta_{\br_t}^2)} \cdot |\bS_{t}|.
\]
\end{itemize}

\end{definition}

Let us give some interpretation of \Cref{def:finefilter}. 
Condition (1) should be thought of as stipulating that the halfspace that is fixed at time $t$, i.e. $w^{(\ind(t))}$, has few points in $\bS_{t+1}$ on the ``wrong side'' (with $w^{(\ind(t))} \cdot x \not = \bs_t$). Condition $(2)$ roughly mandates that the impurity $\imp_j(t)$ of halfpaces ranked higher on the leaderboard doesn't increase by too much after fixing $w^{(\ind(t))}$. For technical reasons, we must use the surrogate quantity $\mimp_j(t)$ instead of $\imp_j(t)$. In particular, if $\imp_j$ is very small, then we may be unable to ensure property $(2)$ with $\imp_j$ in place of $\mimp_j$. While this problem can be circumvented by replacing (2) with
\[ \imp_j(t) \leq \imp_j(t-1) \cdot O(\sqrt{n} \beta_{\br_t}^2 \log^{\qual_j(t)+1}(n) + (\text{small additive quantity}),\]
it is simpler and more convenient to use $\mimp_j$ to bound the multiplicative increase in $(2)$.
Condition (3) should be thought of as stipulating that the set of examples ``that are in play'' does not shrink too much from time step $t$ to time step $t+1$.

With this, we can define a \emph{great total execution} of the loop on Line \ref{line:ltf-func-inner-loop}:

\begin{definition}
We say that a total execution of the loop on Line \ref{line:ltf-func-inner-loop} is \emph{great} if it is a good total execution and moreover, for all $t \in [\log^k(n)]$, the guess $\bg^{(t)}$ is a fine filter. 
\end{definition}

The bulk of the technical work in the rest of this section will be to show that a great total execution occurs with non-trivial probability, i.e.~to establish the following:

\begin{lemma} [Non-tiny probability of great total execution]
\label{lem:great-exec-prob}
    A great total execution of the loop on Line \ref{line:ltf-func-inner-loop} occurs with probability $n^{-O(n \beta_1^2)} \geq 2^{-\sqrt{n} \log^{O(k)}(n)}$.
\end{lemma}

\Cref{sec:advantage-real} and \Cref{sec:filtering} are dedicated to proving \Cref{lem:great-exec-prob}.
In the rest of this subsection we show that given \Cref{lem:great-exec-prob}, we easily get a weak hypothesis as asserted in \Cref{lem:ltf-func-loop}. 

To do this, we first need the following simple lemma.  In the lemma and subsequently, we say that $\bu_j$ is \emph{updated} at time step $t$ to mean that $\bu_j$ takes on a new nonzero value, i.e., $\bu_j(t) \notin\{\bu_j(t-1),0\}$, and
we say that  $\bu_j$ is \emph{reset} to mean that $\bu_j$ is newly set to 0, i.e., 
$\bu_j(t-1) \neq 0, \bu_j(t)=0$.

\begin{lemma}
\label{lem:num-updates}
    In any great total execution of the loop on Line \ref{line:ltf-func-inner-loop}, for any $j \in [k-1]$ and any $i \geq 1$ such that $j + i \leq k$, and any interval $[t_1, t_2]$ of time steps in which $\bu_j$ is not updated or reset, there can be at most $\log^{i}(n)$ time steps in $[t_1,t_2]$ at which $\bu_{j+i}$ is updated or reset.
    Moreover, there are at most $\log(n)$ time steps across any entire great total execution at which $\bu_1$ is updated or reset.
\end{lemma}

\begin{proof}
    We begin by establishing the last sentence of the lemma.  This holds since $\bu_1$ is never reset, and every time it is updated the value of $\qual_1(t)$ must increase (recalling case~$(i)$ of \Cref{def:good}), but recalling the upper bound on $\qual_j(t)$ mentioned after \Cref{def:quality} this can occur at most $\log n$ times.

We prove the first sentence of the lemma by induction on $i$. We begin with the base case, where $i =1$. 
Note that $\bu_{j+1}$ cannot be reset to 0 without updating $\bu_\ell$ for some $\ell \leq j$, and thus updating or resetting $\bu_j$. Since by asumption $\bu_j$ is not updated or reset in the interval $[t_1, t_2]$, it follows that $\bu_{j+1}$ is never reset in that interval.
 On the other hand, each time we update $\bu_{j+1}$ at some time step $t$, by inspection of cases $(i)$ and $(ii)$ of \Cref{def:good}, it must be the case that 
$\qual_{j+1}(t)$
 increases.  Moreover, by \Cref{obs:decreasing-quality} and \Cref{def:quality}, for all $t$ we have that
    \begin{equation} \label{eq:porknachos}
    \log(n) \geq 
    \qual_{j+1}(t) \geq 0.
    \end{equation}
    Thus, $\bu_{j+1}$ can be updated (and thus increase) at most $\log(n)$ times in $[t_1, t_2]$.
    
We now turn to prove the inductive hypothesis. Indeed, suppose the statement is true for $i=\Delta - 1$. This means that, between any two times in $[t_1,t_2]$ that $\bu_{j+1}$ is updated or reset, there can be at most $\log^{\Delta-1}(n)$ times that $\bu_{j+\Delta}$ is updated or reset. Since, by the base case, we can only update or reset $\bu_{j+1}$ at most $\log(n)$ times in $[t_1,t_2]$, it follows that we can update or reset $\bu_{j + \Delta}$ at most $\log^\Delta (n)$ times in $[t_1,t_2]$, and the lemma is proved.
\end{proof}

With this, we can now prove \Cref{lem:ltf-func-loop} (assuming \Cref{lem:great-exec-prob}):

\begin{proof}[Proof of \Cref{lem:ltf-func-loop} assuming \Cref{lem:great-exec-prob}]
    Note that by \Cref{lem:great-exec-prob}, it suffices to show that a great execution will successfully output  a hypothesis $h_{b_1,b_2,b_3}$ that correctly classifies ${\frac 1 2} + 2^{-\sqrt{n} \log^{O(k)}(n)}$ fraction of the examples in $S$. Toward this goal, let $T$ denote the final iteration of the loop, i.e., the first time such that $\bu_i(T) \not = 0$ for all $i \in [k]$ (equivalently, the first time that all $k$ slots on the leaderboard are occupied, i.e.,~at which we have $\br_T=k$). Fix any $j \in [k]$. Since we are considering a great execution $\bg_{\bu_j(T)}$ is a fine filter, and hence by Item $(1)$ of \Cref{def:finefilter}, since $\br_{\bu_{j}(T)} = j$, we have that
    \begin{equation} \label{eq:lambshank}
    \imp_j(\bu_j(T)) \leq e^{-\sqrt{n} \beta_j^2 \log^{\qual_j(T) -1}(n)/3}. 
    \end{equation}
    Moreover, for any $t \geq u_j(T) + 1$, we have that
    \begin{equation} \label{eq:lollipop}
    \mimp_j(t) \leq \mimp_j(t-1) \cdot e^{O(\sqrt{n} \beta_{\br_t}^2 \log^{\qual_j(t) +1}(n))}
    \end{equation}
    by Item~(2) of \Cref{def:finefilter}.
    Using iterated applications of \Cref{eq:lollipop} with $t=\bu_j(T)+1,\dots,T-1$ (and observing that $\qual_j(t)=\qual_j(T)$ for all these values of $t$, the definition of $\qual_j(t)$) followed by \Cref{lem:num-updates} then yields that 
        \begin{align*}
            \mimp_j(T-1) &\leq \mimp_j(\bu_j(T)) \cdot \exp \left( O \left( \sqrt{n} \log^{\qual_j(T)+1}(n) \cdot \sum_{i = 1}^{k-j} \beta_{j+i}^2 \cdot \log^{i}(n) \right) \right). \\
\intertext{By the definition of $\mimp_j(t)$ and \Cref{eq:lambshank}, the above is}
            &= \exp \left(- \sqrt{n} \beta_j^2 \log^{\qual_j(T) - 1}(n)/3 + O \left( \sqrt{n} \log^{\qual_j(T)+1}(n) \cdot \sum_{i = 1}^{k-j} \beta_{j+i}^2 \cdot \log^{i}(n) \right) \right). \\
\intertext{Since the $\beta_j$'s form a geometrically decaying series, the above is}
            &\leq \exp \left(- \sqrt{n} \beta_j^2 \log^{\qual_j(T) - 1}(n)/3 + O \left( \sqrt{n} \beta_{j+1}^2 \log^{\qual_j(T)+2}(n) \right) \right) \\
            &\ll \frac{1}{n}
        \end{align*}
    where the last inequalities hold by our choice of $\beta_i$'s. It then follows by a union bound over all $j \in [k]$ that all but at most a $\frac{k}{n}=o(1)$ fraction of points in $\bS_T$ satisfy 
        \[\sign(w^{(\ind(\bu_j(T)))} \cdot x) = \bs_{\bu_j(T)}\]
    for all $j \in [k]$. This directly implies  that $f$ is at least $(1 -o(1))$-biased towards the bit $b_3 := g(\bs_{\bu_1(T)}, \bs_{\bu_2(T)}, \dots, \bs_{\bu_k(T)})$ on $\bS_T$. For a suitable choice of the bits $b_1$ and $b_2$, then, the hypothesis $h_{b_1,b_2,b_3}$  correctly classifies at least a 
    \begin{align*}
    \frac{1}{2} + \frac{1}{2} \cdot \pbra{1 - o(1)} \cdot \frac{|\bS_T|}{|S|} &\geq \frac{1}{2} + \frac{1}{2} \prod_{t = 1}^T e^{-O(n \beta_{r_t}^2)} \\ 
    &\geq \frac{1}{2} + \frac{1}{2} \prod_{i = 1}^{k} e^{-O(n \beta_{i}^2 \log^i(n))} \\
    &\geq \frac{1}{2} + \frac{1}{2} \cdot e^{-O(n \beta_1^2 \log(n))}
    \end{align*}
    fraction of examples in $S$, where the first inequality above used the fact that every guess $\bg^{(t)}$ is a fine filter (specifically, part (3) of \Cref{def:finefilter}), the second inequality used \Cref{lem:num-updates}, and the third inequality holds because because the $\beta_j$'s form a geometrically decaying series. This is the assertion of \Cref{lem:ltf-func-loop}.
\end{proof}

\subsection{The Advantage Lemma}
\label{sec:advantage-real}

To prove \Cref{lem:great-exec-prob} we will need a more refined version of the ``advantage lemma'' from the warm-up, \Cref{lem:advantage-simplified}.  In this more refined version, stated below, we give an \emph{upper bound} as well as a lower bound on the advantage for a much wider range of parameters, and the upper bound matches the lower bound up to polynomial factors.

\begin{lemma}[Advantage Lemma] \label{lem:advantage}
Let
$\Omega \left(\sqrt{\frac{\log(n)}{n}} \right) 
\leq 10\beta = \alpha
= o(1)$
and let $x,\xref, w \in \mathbb{S}^{n-1}$ satisfy $w \cdot \xref = 0$ and $\alpha \beta (w \cdot x) \geq \omega \left( \frac{\log(n)}{n} \right)$.
If $\bg \sim \calN(0, \frac{1}{n} I_n)$ and $\tbg$ denotes $\bg$ conditioned on it being $\alpha$-lucky for $w$, then 
\[
\Adv(x,\xref, R^\beta_{+}(\tbg)) = 
e^{\Theta(n \alpha \beta (w \cdot  x))}.
\]
\end{lemma}

\begin{proof}
    Without loss of generality let $w = e_1$ and note that as a consequence of this convention and our assumptions, we have that $\tbg_1 \geq \alpha$, $(x_{\mathrm{ref}})_1 = 0,$ and $x_1 \geq 0$.
    Fix some $\gamma \geq \alpha$ and 
    condition on $\tbg_1 = \gamma$. It then follows that $(\tbg_2, \dots, \tbg_n)$ is drawn according to $\calN(0,\frac{1}{n} I_{n-1})$. Now set
    \[t_{\gamma}(z) := \frac{\beta - \gamma z_1}{\sqrt{1 - z_1^2}} \cdot \sqrt{n},
    \quad \quad \text{so} \quad \quad
    t_{\gamma}(\xref) = \beta \sqrt{n}.
    \]
Note that since $(\tbg_2, \dots, \tbg_n) \sim \calN(0,\frac{1}{n} I_{n-1})$, it follows that
\begin{equation}\Prx_{\tbg} \left[x \in R^\beta_+(\tbg) \bigg | \tbg_1 = \gamma \right] = \Prx_{\by \sim \calN(0,1)} [\by \geq t_\gamma(x)] = \begin{cases} \Theta(\frac{1}{t_{\gamma}(x)} e^{-t_{\gamma}(x)^2/2}) & \text{if~}t_{\gamma} (x)\geq 1 \\ \Theta(1) & \text{otherwise} \end{cases}  
\label{eq:carrot}
\end{equation} 
and likewise
\begin{equation}\Prx_{\tbg} \left[\xref \in R^\beta_+(\tbg)  \right] = \Prx_{\by \sim \calN(0,1)} [\by \geq t_\gamma(\xref)] = \Theta\pbra{\frac{1}{\beta \sqrt{n}} e^{-n\beta^2/2}}
\label{eq:beet}
\end{equation} 
where in both cases the final equality is by \Cref{lem:gauss-lb} (note that in \Cref{eq:beet} we know that $t_\gamma(\xref) = \beta \sqrt{n}$ is $>1$ by assumption on $\beta$).

With this we turn to prove the lower bound on the advantage (analogous to what was done in the warm-up):

\begin{claim}
    \label{clm:adv-lb}
    \[
\Adv(x,\xref, R^\beta_{+}(\tbg)) = 
e^{\Omega(n \alpha \beta (w \cdot  x))}.
\]
\end{claim}

\begin{proof}
The argument is similar to the proof of \Cref{lem:advantage-simplified}.
Since the Gaussian tail $\Pr_{\by \sim {\cal N}(0,1)}[\by \geq t]$ is a decreasing function of $t$ and $t_\gamma(x)$ is decreasing in $\gamma$, the probability given in \Cref{eq:carrot} is increasing in $\gamma$. 
Thus, for each $\gamma \geq \alpha$, we have that

\[\Prx_{\tbg} \left[x \in R^\beta_+(\tbg) \bigg | \tbg_1 = \gamma \right] \geq \begin{cases} \Omega(\frac{1}{t_{\alpha}(x)} e^{-t_{\alpha}(x)^2/2}) & \text{if~}t_{\alpha}(x) \geq 1 \\ \Theta(1) & \text{otherwise} \end{cases}. \]
Since $\tbg$ is a mixture over outcomes with $\tbg_1=\gamma$ as $\gamma$ ranges over $[\alpha,\infty)$, we have
\begin{equation} \label{eq:parsnip}
\Prx_{\tbg} \left[x \in R^\beta_+(\tbg)  \right] \geq \begin{cases} \Omega(\frac{1}{t_{\alpha}(x)} e^{-t_{\alpha}(x)^2/2}) & \text{if~}t_{\alpha}(x) \geq 1 \\ \Theta(1) & \text{otherwise} \end{cases}.
\end{equation}
We now consider two cases depending on the value of $t_\alpha(x).$

The first case is that $t_{\alpha}(x) \leq 1$; in this case, recalling \Cref{def:advantage}, \Cref{eq:beet,eq:parsnip} give us that
\[ \Adv(x,\xref, R^\beta_{+}(\tbg)) \geq \Omega \left( \beta \sqrt{n} e^{n\beta^2/2} \right) \geq e^{\Omega(n\beta^2)}
\geq e^{\Omega(n \alpha \beta (w \cdot x))},\]
where the penultimate inequality used that $\beta \sqrt{n} \geq 1$
and the final inequality used that $n \alpha \beta(w \cdot x) = O(n \beta^2),$ which holds since $\alpha =10\beta$ and $w \cdot x \leq 1.$

The second case is that $t_{\alpha}(x) \geq 1$. In this case, again by \Cref{eq:beet,eq:parsnip} and using $\beta \sqrt{n} \geq 1,$ we have 
\begin{equation}
\Adv(x,\xref, R^\beta_{+}(\tbg)) = \Omega \left( \frac{\beta \sqrt{n}}{t_\alpha(x)} e^{-t_{\alpha}(x)^2/2 + n\beta^2/2}\right).
\label{eq:sugakix}
\end{equation}
We can then compute
\begin{align*}
-t_\alpha(x)^2/2 + n\beta^2/2 &= \frac{n}{2} \left( \frac{2 \alpha \beta x_1 -\alpha^2 x_1^2 - \beta^2 x_1^2}{1 - x_1^2} \right). 
\end{align*}
We then have
\begin{align*}-t_\alpha(x)^2/2 + n\beta^2/2 &= {\frac n 2} \cdot {\frac 1 {1-x_1^2}} \left(  2 \alpha \beta x_1 -\alpha^2 x_1^2 - \beta^2 x_1^2 \right) \\
&\geq {\frac n 2}
\cdot {\frac 1 {1-x_1^2}} 
(\alpha \beta x_1 - \beta^2 x_1^2)
\tag{using $\beta \geq \alpha x_1$, since $t_{\alpha}(x) > 0$}\\
&\geq {\frac n 2}
\cdot {\frac 1 {1-x_1^2}} 
(\alpha \beta x_1 / 2) \tag{using $2\beta \leq \alpha$ and $0 \leq x_1 \leq 1$}\\
&= \Omega(n \alpha \beta x_1).
\end{align*}

As $t_\alpha(x) \geq 0$ and $2 \beta \leq \alpha$, it follows that $x_1 \leq \frac{1}{2}$. Since $\beta - \alpha x_1 \leq 1$ and $0 \leq x_1 \leq \frac{1}{2}$, it follows that
\begin{equation}t_{\alpha}(x) \leq O(\sqrt{n}). 
\label{eq:crustymunchies}
\end{equation}
Thus, in this second case, combining \Cref{eq:sugakix,eq:crustymunchies} we have that 
\[ \Adv(x,\xref, R^\beta_{+}(\tbg)) \geq \Omega \left( \beta e^{\Omega(n \alpha \beta x_1)} \right) \geq e^{\Omega(n \alpha \beta x_1)}
=e^{\Omega(n \alpha \beta (w \cdot x))},\]
where the final inequality uses the fact that $\beta = \Omega\pbra{\sqrt{\frac {\log(n)} n}}$ and $\alpha \beta x_1 \geq \omega \left( \frac{\log(n)}{n} \right)$.
\end{proof}

It now remains to prove the corresponding upper bound on the advantage (we remark that this part of the argument does not have an analogue in the warm-up). 
We first note that the advantage is never greater than $e^{O(n\beta^2)}$, as
    \begin{equation}
\Adv(x,\xref, R^\beta_{+}(\tbg)) 
    = 
    \frac{\Pr[x \in R^\beta_{+}(\tbg)]}{\Pr[\xref \in R^\beta_{+}(\tbg)]}
    \leq \frac{1}{\Pr[\xref \in R^\beta_{+}(\tbg)]} \leq \beta \sqrt{n} e^{n\beta^2/2}
    = e^{O(n\beta^2)}
    \label{eq:zorba}
    \end{equation}
    where the final inequality is \Cref{eq:beet} and the final equality uses $\beta \sqrt{n} = \Omega(\sqrt{\log n})$.
    If $\alpha x_1 \geq \beta/4$ then $n \alpha \beta x_1 \geq n \beta^2/4$ and \Cref{eq:zorba} gives the desired upper bound on advantage; so let us assume that $\alpha x_1 \leq \beta/4$. Now note that by \Cref{lem:gauss-lb} we have
    \begin{equation}
        \label{eq:ummagumma}
    \Pr[\tbg_1 \geq 2 \alpha] \leq \frac{\frac{\Theta(1)}{2\alpha \sqrt{n}} e^{-2 n \alpha^2}}{\frac{\Theta(1)}{\alpha \sqrt{n}} e^{-n \alpha^2/2}} = \Theta \left( e^{-1.5 n \alpha^2} \right). 
    \end{equation}
Moreover, note that
    \begin{equation} \label{eq:tbiggerthanone}t_{2 \alpha}(x) = \frac{\beta - 2 \alpha x_1}{\sqrt{1 - x_1^2}} \cdot \sqrt{n} \geq \frac{\beta \sqrt{n}}{2} \geq 1.\end{equation}
We thus have
    \begin{align}
\Prx_{\tbg} \left[x \in R^\beta_+(\tbg) \right] &=
\Prx_{\tbg} \left[x \in R^\beta_+(\tbg)  \ \& \ \tbg \geq 2 \alpha \right ]
+
\Prx_{\tbg} \left[x \in R^\beta_+(\tbg) 
 \ \& \ \tbg < 2 \alpha \right ]\nonumber
\\        
&\leq \Theta \left( e^{-1.5 n \alpha^2} \right) + \Prx_{\tbg} \left[x \in R^\beta_+(\tbg) 
 \ \& \ \tbg < 2 \alpha \right ] \nonumber \tag{by \Cref{eq:ummagumma}}\\
&\leq \Theta \left( e^{-1.5 n \alpha^2} \right) + \frac{1}{t_{2\alpha}(x)} e^{-t_{2\alpha}(x)^2/2}, \label{eq:ace}
    \end{align}
    where the last inequality, as at the start of the proof of \Cref{clm:adv-lb}, uses the fact that the probability given in \Cref{eq:carrot} is increasing in $\gamma$.
Similar to before, we then compute
\[
-t_{2\alpha}(x)^2/2 + n\beta^2/2 = \frac{n}{2} \left( \frac{4 \alpha \beta x_1 -4\alpha^2 x_1^2 - \beta^2 x_1^2}{1 - x_1^2} \right).
\]
As before we must have that $x_1 \leq 1/2$ since $2\beta < \alpha$ and $t_{2 \alpha}(x) \geq 1$. Similar to before, it then follows that
\begin{equation}
-t_{2\alpha}(x)^2/2 + n\beta^2/2 = O(n \alpha \beta x_1) 
\label{eq:almost}
\end{equation}
where we now use the fact that $\alpha x_1 \leq \beta/4$. Thus, we can bound the overall advantage by 
\begin{align*}
    \Adv(x,\xref, R^\beta_{+}(\tbg)) 
    &=
    \frac{\Pr[x \in R^\beta_{+}(\tbg)]}{\Pr[\xref \in R^\beta_{+}(\tbg)]}\\
    &\leq O \left(\beta \sqrt{n} e^{-1.5n\alpha^2 + n\beta^2/2} \right)  +
 \frac{\beta \sqrt{n}}{t_{2\alpha}(x) } e^{-t_{2\alpha}(x)^2/2 + n \beta^2/2} \tag{by \Cref{eq:beet,eq:ace}}\\
 &\leq O \left(\beta \sqrt{n} e^{-n\beta^2} \right) 
 + \beta\sqrt{n} e^{O(n \alpha \beta x_1)} \tag{using $\alpha \geq \beta$ and  \Cref{eq:tbiggerthanone,eq:almost}}\\
 &\leq 
 e^{O( n \alpha \beta x_1 )},
    \end{align*}
where the last line again uses $\alpha \beta x_1 \geq \Omega \left( \frac{\log(n)}{n} \right)$.
\end{proof}

Next, we establish a ``monotonicity'' property of advantage; this essentially says that if $x$ points more in the direction of $w$ than $\xref$, then $x$ has advantage over $\xref$ (for a ``lucky'' $\tbg$ distributed as in the previous lemma).

\begin{lemma}[Monotonicity of Advantage] \label{lem:monotonicity}
Let $\beta \leq \alpha$ and $x,\xref,w \in \mathbb{S}^{n-1}$ be such that $w\cdot x \geq $ $w \cdot  \xref$. Then
    \[ \Adv(x,\xref, R^\beta_{+}(\tbg)) \geq 1\]
 where $\tbg$ denotes $\bg \sim \calN(0,\frac{1}{n} I_n)$ conditioned on its being $\alpha$-lucky for $w$.
 \end{lemma}

\begin{proof}
As before, without loss of generality let $w = e_1$ and note that by assumption we have that $\tbg_1 \geq \alpha$. Fix some $\gamma \geq \alpha$ and condition on $\tbg_1 = \gamma$. As before, we have that $(\tbg_2, \dots, \tbg_n)$ is drawn according to $\calN(0,\frac{1}{n} I_{n-1})$.
As before, for $z$ an arbitrary unit vector, let
    \[t_{\gamma}(z) := \frac{\beta - \gamma z_1}{\sqrt{1 - z_1^2}} \cdot \sqrt{n}. \]
As before, since $(\tbg_2, \dots, \tbg_n) \sim \calN(0,\frac{1}{n} I_{n-1})$, it follows from \Cref{lem:gauss-lb} that for any unit vector $z$ we have
\[\Prx_{\tbg} \left[z \in R^\beta_+(\tbg) \bigg | \tbg_1 = \gamma \right] = \Prx_{\by \sim \calN(0,1)} [\by \geq t_\gamma(z)]. \] 

We now note that 
    \[\frac{d}{dz_1} t_{\gamma}(z) = \sqrt{n} \cdot \frac{\beta z_1 - \gamma}{(1 - z_1^2)^{3/2}} \leq 0,\]
where the final inequality follows because $z_1 \leq 1$ and as $\beta \leq \alpha \leq \gamma$. So, we can conclude that
\[\Prx_{\tbg} \left[\xref \in R^\beta_+(\tbg) \bigg | \tbg_1 = \gamma \right] =  \Prx_{\by \sim \calN(0,1)} [\by \geq t_\gamma(\xref)] \leq \Prx_{\by \sim \calN(0,1)} [\by \geq t_\gamma(x)] = \Prx_{\tbg} \left[x \in R^\beta_+(\tbg) \bigg | \tbg_1 = \gamma \right].\] 

Writing $\phi(\gamma)$ to denote the value of the pdf of $\tilde{\bg}_1$ at $\gamma$, it then follows that
\begin{align*}
    \Adv(x,\xref, R^\beta_{+}(\tbg)) &= \frac{\int_\alpha^\infty \Pr_{\tbg} \left[x \in R^\beta_+(\tbg) \bigg | \tbg_1 = \gamma \right] \cdot 
    \phi(\gamma)
    d\gamma}{\int_\alpha^\infty \Pr_{\tbg} \left[\xref \in R^\beta_+(\tbg) \bigg | \tbg_1 = \gamma \right] \cdot 
    \phi(\gamma)
    d\gamma} \\
    &\geq \min_{\gamma \geq \alpha} \frac{\Pr_{\tbg} \left[x \in R^\beta_+(\tbg) \bigg | \tbg_1 = \gamma \right]}{\Pr_{\tbg} \left[\xref \in R^\beta_+(\tbg) \bigg | \tbg_1 = \gamma \right]} \\
    & \geq 1.\qedhere
\end{align*}
\end{proof}

\subsection{The Filtering Lemma and the Proof of \Cref{lem:great-exec-prob}} \label{sec:filtering}

The main goal of this subsection is to prove the following ``filtering lemma'':

\begin{lemma}[Filtering Lemma] \label{lem:filtering-real}
    Consider an execution of the loop on Line \ref{line:ltf-func-inner-loop} and fix some time step $t$. Then
        \begin{align}
        \label{eq:hippomeat}
        &\Pr_{\bg^{(t)}} \sbra{[\bg^{(t)} \text{ is a fine filter~} \bigg | ~~ \parbox{15em}{$t'$ is good for all $t' \leq t$ and\\ $\bg^{(t')}$ is a fine filter for all $t' < t$} } \geq e^{-O(n \beta_{r_t}^2)}. 
        \end{align}
\end{lemma}

Note that because of the conditioning on the LHS of \Cref{eq:hippomeat}, the value of $\br_t$ has been fixed after this conditioning. This can be verified by sequentially considering $t'=1,2,\dots$ in items $(i)$ and $(ii)$ of \Cref{def:good}:  at each time step $t'$, given the values of the earlier $r_{t''}$'s (which determine the earlier $\bu_{j}(t'')$'s), the value of $\br_{t'}$ is a deterministic function of the target halfspaces, and therefore so are the $u_{j}(t')$'s, letting us continue the sequential argument.  This is why the $r_t$ on the RHS of \Cref{eq:hippomeat} is well defined and not random (and why the $r_t$ occurring there is no longer boldfaced).

Before proving \Cref{lem:filtering-real}, we show why \Cref{lem:filtering-real} completes the proof of \Cref{lem:great-exec-prob}:

\begin{proof}[Proof of \Cref{lem:great-exec-prob} using \Cref{lem:filtering-real}]
    Fix a time step $t$. Let us suppose that $t'$ is good and $\bg^{(t')}$ is a fine filter for all $t' < t$. We will lower bound the probability that $t$ is good; once we have done this, we will deploy \Cref{lem:filtering-real} to lower bound the probably that $\bg^{(t)}$ is a fine filter. This will allow us to inductively lower bound the probability of a great total execution (recall that this means that all time steps are good and all guesses $\bg^{(t)}$ are fine filters).  
    
    This is easily accomplished:  note that the probability that time step~$t$ is good (meaning that $\bs_t$ is correct, $\br_t$ takes the correct value, 
    and $\bg^{(t)}$ is $\alpha_{\br_t}$-lucky with respect to the appropriate halfspace $w^{\ind(u_{r_t})}$, as specified in \Cref{def:good}) is 
        \[\frac{1}{2} \cdot \frac{1}{k} \cdot e^{-O(n \alpha_{r_t}^2)} = e^{-O(n \beta_{r_t}^2)} \]
        (recalling that $\beta_j = \alpha_j/10$ for all $j$).
        Assuming this happens, we can apply the Filtering Lemma  (\Cref{lem:filtering-real}) to get that there is a $e^{-O(n \beta_{r_t}^2)}$ chance that $\bg^{(t)}$ is a fine filter. 
        
        Using this inductively for $t=1,\dots,T$, where
        $T$ denotes the final time step, we can conclude 
        that the probability of a great total execution is at least
        \[\prod_{t = 1}^{T} e^{-O(n \beta_{r_t}^2)} \geq  \exp \left( - n \sum_{i = 1}^k \log(n)^i \beta_i^2 \right) \geq \exp \left( - n \beta_1^2 \log(n) \right)\]
    as desired, where the first inequality is by \Cref{lem:num-updates} and the second is by the geometrically decreasing definition of the $\beta_i$'s.
\end{proof}

We now move on to proving \Cref{lem:filtering-real}. Its proof (most of which takes place in the proof of \Cref{claim:expectation-lb}) consists of three components, corresponding to the three properties of \Cref{def:finefilter}. Recall that property~(1) bounds the impurity of the set of points $\bS_{t+1}$ that are in play at the end of time step $t$, and property~(3) states that that set is not too small.  These are entirely analogous to corresponding properties in our earlier warm-up argument, and proving properties (1) and (3) of \Cref{def:finefilter} indeed follows via similar arguments to those in \Cref{sec:warmup}. To prove property $(2)$, which states that the modified impurity doesn't increase by too much, we partition points based on their margin with respect to $w^{(\ind(t))}$, the halfspace we are currently fixing. To handle points close to the halfspace, we can simply apply the upper bound from the Advantage Lemma. For the remaining points, on the other hand, we have a bound on the quality for the halfspace that we are currently fixing by virtue of where it is placed on the leaderboard. From this, we have that there are ``few'' points with far margin. Using this with the upper bound on advantage from the Advantage Lemma then allows us to bound the contribution of these far margin points to $\mimp_j(t)$.

To prove \Cref{lem:filtering-real}, we will use the following key claim, which is roughly analogous to \Cref{eq:cheesepuff} in the proof of \Cref{lem:filtering-simple}.

    \begin{claim}
    \label{claim:expectation-lb}
    Fix any $t$, and suppose that $\bg^{(t)}$ is drawn conditioned on $t'$ being good for all $t' \leq t$ and $\bg^{(t')}$ being a fine filter for all $t' < t$.  
    Fix an outcome $S_t$ of $\bS_t$, and
    define the random variable $\bX$ to be
    \begin{equation} \label{eq:X}
    \bX := |\bS_{t+1}| -\overbrace{e^{\sqrt{n} \beta_{r_t}^2 \log^{\qual_{r_t}(t)-1}(n)/2.5} \cdot \imp_{r_t}(t) |\bS_{t+1}|}^{:=\bU}  - \sum_{j=1}^{r_t-1} \overbrace{\frac{\imp_j(t) |\bS_{t+1}|}{\mimp_j(t-1) \cdot e^{O(\sqrt{n} \beta_{r_t}^2 \log^{\qual_j(t)+1}(n))}}}^{:=\bV_j} .
    \end{equation}
    Then
    \begin{align*}
        \Ex_{\bg^{(t)}} \sbra{ \bX}
        \geq |S_t| e^{-O(\sqrt{n} \beta_{r_t}^2)}.
    \end{align*}
    \end{claim}

\begin{proofof}{\Cref{lem:filtering-real} assuming \Cref{claim:expectation-lb}}
Fix any outcome $S_t$ of $\bS_t$, and observe that the random variable $\bX$ can never exceed $|S_t|$ (since $\bS_{t+1} \subseteq S_t$).  Hence by \Cref{lem:reverse-Markov} (``reverse Markov''), with probability at least $e^{-O(\sqrt{n} \beta_{r_t}^2)}$ we have that
    \begin{equation} 
    \label{eq:japanese-mayonnaise}\bX \geq e^{-O(\sqrt{n} \beta_{r_t}^2)}|S_t| . \end{equation}
Since moreover $\bX \leq |\bS_{t+1}|$, this immediately implies that 
 \begin{equation}
     \label{eq:fujiapple}
 |\bS_{t+1}| \geq e^{-O(\sqrt{n} \beta_{r_t}^2)} |S_t|
 \end{equation}
 (note that this corresponds to item~(3) of \Cref{def:finefilter}).
Using \Cref{eq:X} and \Cref{eq:fujiapple}, we can rewrite \Cref{eq:japanese-mayonnaise} as 
\[
2|\bS_{t+1}| \geq |\bS_{t+1}| + e^{-O(\sqrt{n}\beta_{r_t}^2)} |S_t|  \geq  \bU + \bV,
\text{~and hence~}2|\bS_{t+1}| \geq \bU,\bV,
\]
where $\bV := \sum_{j=1}^{r_t-1} \bV_j.$
The inequality $2|\bS_{t+1}| \geq \bU$ can be rewritten as
\[
\imp_{r_t}(t) \leq 2e^{-\sqrt{n} \beta_{r_t}^2 \log^{\qual_{r_t}(t)-1}(n)/2.5}
< e^{-\sqrt{n} \beta_{r_t}^2 \log^{\qual_{r_t}(t)-1}(n)/3}
\]
(corresponding to item~(1) of \Cref{def:finefilter}),
and the inequality $2|\bS_{t+1}| \geq \bV$ implies that for each $j < r_t,$ we have
\[ \imp_j(t) \leq \mimp_j(t-1) e^{O(\sqrt{n} \beta_{r_t}^2 \log^{\qual_j(t)+1}(n))}.\]
Note now that $\qual_j(t) = \qual_j(t-1)$ as $j < r_t$. Thus it follows that \[e^{-\sqrt{n} \beta_j^2 \log^{\qual_j(t) - 1}(n)/3} = e^{-\sqrt{n} \beta_j^2 \log^{\qual_j(t-1) - 1}(n)/3} \leq \mimp_j(t-1)\]
 by the definition of $\mimp_j(t-1)$.  Combining this with the above upper bound on $\imp_{r_t}(t)$ then yields item~(2) of \Cref{def:finefilter}.
Thus we have indeed established that $\bg^{(t)}$ is a fine filter, as desired, for that outcome of $S_t$.  Since this holds for every outcome $S_t$ of $\bS_t$, the proof of \Cref{lem:filtering-real} is complete.
\end{proofof}

\begin{proof}[Proof of \Cref{claim:expectation-lb}]
Fix a $j < r_t$; the idea is to ``bucket'' the points in $S_t$ that are on the wrong side of the $w^{(\ind(\bu_j(t)))}$ halfspace by their margin. To do this, we define 
    \[ A^j := \left \{x \in S_{t}: \sgn(w^{(\ind(\bu_j(t)))} \cdot x) \not = \bs_{\bu_j(t)} \land  \left| w^{(\ind(\bu_j(t)))} \cdot x \right| \leq \frac{\log^{\qual_j(t)+1}(n)}{2\sqrt{n}} \right \} \]
(intuitively, these are the ``low-margin'' points on the wrong side of the $w^{(\ind(\bu_j(t)))}$ halfspace in $S_t$). Additionally, for $i > \qual_j(t)$, we define a collection of mutually disjoint subsets
    \[B_i^j := \left \{x \in S_{t}: \sgn(w^{(\ind(\bu_j(t)))} \cdot x) \not = \bs_{\bu_j(t)} \land  \left| w^{(\ind(\bu_j(t)))} \cdot x \right| \in \left( \frac{\log^i(n)}{2\sqrt{n}}, \frac{\log^{i+1}(n)}{2\sqrt{n}} \right] \right \}, \]
    based on how large the margin is for the ``high margin'' points on the wrong side  of the $w^{(\ind(\bu_j(t)))}$ halfspace in $S_t$.
    (We remark that these sets are akin to the $T_\lambda$ sets from the proof overview given in \Cref{sec:learnfunc}.)
Notice that the set of points in $S_t$ on the wrong side of the $w^{(\ind(\bu_j(t)))}$ halfspace is exactly $A^j \sqcup \pbra{\bigsqcup_{i=\qual_j(t)+1}^{\log n} B_i^j}$.

Let $\xref$ denote any point such that $w^{(\ind(\bu_j(t)))} \cdot \xref = 0$ (we remark that here we are considering $\xref$ only for the purposes of analysis; it is not an actual point in the $S_t$).  We can now compute that 
    \begin{align}
        \Ex_{\bg^{(t)}} [ |A^j \cap \bS_{t+1}|] 
        &=\sum_{x \in A^j} \Pr_{\bg^{(t)}} \left[ x \in \bS_{t+1} \right] \nonumber\\
        &\leq \sum_{x \in A^j} \Pr_{\bg^{(t)}} \left[ x \in R_{\bs_t}^{\beta_{r_t}}(\bg^{(t)}) \right] \tag{recalling Line~\ref{line:forsterize-again}} \\
        &\leq \sum_{x \in A^j} e^{O(\sqrt{n} \beta_{r_t}^2 \log^{\qual_j(t)+1}(n))} \Pr \left[\xref \in R_{\bs_t}^{\beta_{r_t}}(\bg^{(t)}) \right] \nonumber \\
        &\leq \imp_j(t-1) \cdot |S_t| \cdot e^{O(\sqrt{n} \beta_{r_t}^2 \log^{\qual_j(t)+1}(n))} \Pr \left[\xref \in R_{\bs_t}^{\beta_{r_t}}(\bg^{(t)}) \right] \label{eq:water}
    \end{align}
    where the second inequality follows from the Advantage Lemma (\Cref{lem:advantage}) and the definition of $A^j$, and the third inequality is because $\imp_j(t-1) \cdot |S_t|$ is an upper bound on the size of $A^j$, recalling the definition of $\imp_j(t-1)$.

    For any $i > \qual_j(t)$, we can similarly compute 
    \begin{align}
        \Ex_{\bg^{(t)}} [ |B_i^j \cap \bS_{t+1}| ] &\leq \sum_{x \in B_i^j}  \Pr_{\bg^{(t)}} \left[ x \in R_{\bs_t}^{\beta_{r_t}}(\bg^{(t)}) \right] \nonumber \\
        &\leq \sum_{x \in B_i^j} e^{O(\sqrt{n} \beta_{r_t}^2 \log^{i+1}(n))} \Pr_{\bg^{(t)}} \left[ \xref \in R_{\bs_t}^{\beta_{r_t}}(\bg^{(t)}) \right] \nonumber
    \intertext{where the second inequality is by the definition of $B^j_i$ and the Advantage Lemma (\Cref{lem:advantage}). 
    Now, we observe that since time step $t$ is good, by \Cref{def:good} the halfspace $w^{(\ind(t))}$ is placed into spot $r_t$ on the leaderboard; since  $j<r_t$ (by the assumption at the start of the proof), it must be the case that
    $\qual_j(w^{(\ind(t))}, t) \leq \qual_j(t) < i$. Recalling the definition $\qual_j(w^{(\ind(t))}, t) $ (\Cref{def:quality}) and the definition of $B_i^j$, it follows that $|B_i^j| < \frac{1}{4n} e^{-\sqrt{n} \beta_j^2 \log^{i-1}(n) / 2 } \cdot |S_t|$, using which we can then bound the above by}
    & \leq \frac{1}{4n} e^{-\sqrt{n} \beta_j^2 \log^{i-1}(n) / 2 } \cdot |S_t|  e^{O(\sqrt{n} \beta_{r_t}^2 \log^{i+1}(n))} \Pr_{\bg^{(t)}} \left[ \xref \in R_{\bs_t}^{\beta_{r_t}}(\bg^{(t)}) \right] \nonumber \\
    & \leq e^{-\sqrt{n} \beta_j^2 \log^{i-1}(n) / 2.2 } |S_t|  \Pr_{\bg^{(t)}} \left[ \xref \in R_{\bs_t}^{\beta_{r_t}}(\bg^{(t)}) \right] \label{eq:wbter}
    \end{align}
where we used the fact that $\beta_j \geq \log^5(n) \beta_{r_t}$ since $j < r_t$ (recalling the definition of the $\beta_j$ parameters).

We now have that
    \begin{align*}
    \Ex[\bV_j] &= 
        \Ex_{\bg^{(t)}} \Bigg[ \frac{\imp_j(t) |\bS_{t+1}|}{\mimp_j(t-1) \cdot e^{O(\sqrt{n} \beta_{r_t}^2 \log^{\qual_j(t)+1}(n))}}  \Bigg] \tag{definition of $\bV_j$} \\
        &= \frac{1}{\mimp_j(t-1) \cdot e^{O(\sqrt{n} \beta_{r_t}^2 \log^{\qual_j(t)+1}(n))}}  \cdot \Ex_{\bg^{(t)}} \Bigg[  \left|A^j \cap \bS_{t+1} \right| + \sum_{i > \qual_j(t)} \left|B_i^j \cap \bS_{t+1} \right| \Bigg] \tag{definition of $A^j$ and $B^j_i$}\\
        &\leq \frac{\Pr \left[\xref \in R_{\bs_t}^{\beta_{r_t}}(\bg^{(t)}) \right] |S_t|}{\mimp_j(t-1) \cdot e^{O(\sqrt{n} \beta_{r_t}^2 \log^{\qual_j(t)+1}(n))}} \cdot\\
        & \text{~~~} \left ( \imp_j(t-1) \cdot e^{O(\sqrt{n} \beta_{r_t}^2 \log^{\qual_j(t)+1}(n))} + \sum_{i > \qual_j(t)} e^{-\sqrt{n} \beta_j^2 \log^{i-1}(n) / 2.2 } \right) 
        \tag{\Cref{eq:water} and \Cref{eq:wbter}}\\
        &\leq \frac{\Pr \left[\xref \in R_{\bs_t}^{\beta_{r_t}}(\bg^{(t)}) \right] |S_t|}{\mimp_j(t-1) \cdot e^{O(\sqrt{n} \beta_{r_t}^2 \log^{\qual_j(t)+1}(n))}} \cdot 
        \\
        & \text{~~~}\left ( \imp_j(t-1) \cdot e^{O(\sqrt{n} \beta_{r_t}^2 \log^{\qual_j(t)+1}(n))} + e^{-\sqrt{n} \beta_j^2 \log^{\qual_j(t)}(n) / 2.3 } \right) \tag{decaying series bound}\\
        &\leq \frac{\Pr \left[\xref \in R_{\bs_t}^{\beta_{r_t}}(\bg^{(t)}) \right] |S_t|}{2k \cdot \mimp_j(t-1)} \cdot 
        \left ( \imp_j(t-1) + e^{-\sqrt{n} \beta_j^2 \log^{\qual_j(t)}(n) / 2.4 } \right) \tag{dividing numerator and denominator by $e^{O(\sqrt{n} \beta_{r_t}^2 \log^{\qual_j(t)+1}(n))}$}\\
        &\leq \frac{\Pr \left[\xref \in R_{\bs_t}^{\beta_{r_t}}(\bg^{(t)}) \right] |S_t|}{k},
        \tag{definition of $\mimp$ in \Cref{def:finefilter}}
    \end{align*}
where we now reveal to the reader that in the definition of $\bV_j$ (recall \Cref{eq:X}), the constant hidden by the big-Oh notation is chosen so as to ensure that the $e^{O(\sqrt{n} \beta_{r_t}^2 \log^{\qual_j(t)+1}(n))}$ quantity we divided by to obtain the penultimate inequality cancels out the $e^{O(\sqrt{n} \beta_{r_t}^2 \log^{\qual_j(t)+1}(n))}$ multiplier of $\imp_j(t-1)$ in the previous line with a factor of $2k$ to spare.

We now turn to analyzing the quantity $\bU$.
To do this, we observe that by monotonicity of advantage (\Cref{lem:monotonicity}) we have that
\begin{align*}
    \E_{\bg^{(t)}} \left[ \imp_{r_t}(t) \cdot |\bS_{t+1}| \right] = \sum_{x \in S_t: w^{(\ind(t))} \cdot x \not = \bs_t} \Pr_{\bg^{(t)}} \left[ x \in R_{\bs_t}^{\beta_{r_t}}(\bg^{(t)}) \right] 
    \leq |S_t| \Pr_{\bg^{(t)}} \left[ \xref \in R_{\bs_t}^{\beta_{r_t}}(\bg^{(t)}) \right],
\end{align*}
so recalling the definition of $\bU$ from \Cref{eq:X}, we have
\[
\E[\bU] = 
e^{\sqrt{n} \beta_{r_t}^2 \log^{\qual_{r_t}(t)-1}(n)/2.5} \cdot  \E[\imp_{r_t}(t) |\bS_{t+1}|]
\leq
e^{\sqrt{n} \beta_{r_t}^2 \log^{\qual_{r_t}(t)-1}(n)/2.5} \cdot
|S_t| \Pr_{\bg^{(t)}} \left[ \xref \in R_{\bs_t}^{\beta_{r_t}}(\bg^{(t)}) \right].
\]

Finally we turn to the first term in \Cref{eq:X}, namely $|\bS_{t+1}|$: we compute
    \begin{align}
    \E[|\bS_{t+1}|] &\geq \frac{1}{n} \E_{\bg^{(t)}} \left[ \left |S_t \cap R_{\bs_t}^{\beta_{r_t}(\bg^{(t)})} \right| \right ] \nonumber\\
    & \geq \frac{1}{n} \cdot \frac{1}{4n} e^{-\sqrt{n} \beta_{r_t}^2 \log^{\qual_{r_t}(t) - 1}(n)/2} |S_t| \cdot e^{\Omega(\sqrt{n} \beta_{r_t}^2 \log^{\qual_{r_t}(t)}(n))} \Pr \left[ \xref \in  R_{\bs_t}^{\beta_{r_t}}(\bg^{(t)}) \right]\nonumber \\
    &\geq e^{\Omega(\sqrt{n} \beta_{r_t}^2 \log^{\qual_{r_t}(t)} (n))} |S_t| \Pr \left[ \xref \in  R_{\bs_t}^{\beta_{r_t}}(\bg^{(t)}) \right], \label{eq:climax} 
    \end{align}
where the first inequality used that the Forster transform can discard at most $\frac{n-1}{n}$ of the points (by \Cref{thm:forster}). For the second inequality, we note by the definition of quality that there is at least a $p = \frac{1}{4n} \exp(-n \beta_{r_t}^2 \tau/\log(n))$ fraction of points with margin $\tau = \frac{\log^{\qual_{r_t}(t)}(n)}{2 \sqrt{n}}$ with respect to $w^{(\ind(t))}$ in $\bS_t$. Combining this with the lower bound from the Advantage Lemma (\Cref{lem:advantage}) then yields this second inequality.

We can now put together \Cref{eq:climax} with our upper bounds on $\E[\bV]$ and $\E[\bU]$ to get that $\E[\bX]$ (recall that this equals $\E[|\bS_{t+1}| - \bU - \bV]$) is equal to
\begin{align*}
        \Ex_{\bg^{(t)}} \Bigg [ &|\bS_{t+1}| -e^{\sqrt{n} \beta_{r_t}^2 \log^{\qual_{r_t}(t)-1}(n)/2.5} \cdot \imp_{r_t}(t) |\bS_{t+1}|  - \sum_{j=1}^{r_t-1} \frac{\imp_j(t) |\bS_{t+1}|}{\mimp_j(t) \cdot e^{O(\sqrt{n} \beta_{r_t}^2 \log^{\qual_j(t)+1}(n))}} \Bigg] \\
        &\geq \Pr \left[ \xref \in R_{\bs_t}^{\beta_{r_t}}(\bg^{(t)}) \right] |S_t| \left(\overbrace{e^{\Omega(\sqrt{n} \beta_{r_t}^2 \log^{\qual_{r_t}(t)} (n))}}^{\text{from~}|\bS_{t+1}|} - \overbrace{e^{\sqrt{n} \beta_{r_t}^2 \log^{\qual_{r_t}(t)-1}(n)/2.5}}^{\text{from~}\bU} - \sum_{j=1}^{r_t-1} \overbrace{{\frac 1 k}}^{\text{from~}\bV_j} \right) \\
        &\geq \Pr \left[ \xref \in R_{\bs_t}^{\beta_{r_t}}(\bg^{(t)}) \right] |S_t| \\
        &\geq e^{-O(n \beta_{r_t}^2)} |S_t|,
    \end{align*}
where the final inequality used \Cref{lem:gauss-lb}.
This completes the proof of \Cref{claim:expectation-lb}.
\end{proof}

\subsection{From weak learning over a sample to strong PAC learning} \label{sec:strong-PAC-any-of-k}

\Cref{lem:ltf-func-loop} gives us 
that with probability at least $2^{-\sqrt{n} \log^{O(k)}(n)}$, each execution of the entire inner Line \ref{line:ltf-func-inner-loop} loop of \Cref{alg:ltf-func} outputs a hypothesis $h_{b_1,b_2,b_3}$ that correctly classifies ${\frac 1 2} + 2^{-\sqrt{n} \log^{O(k)}(n)}$ fraction of the examples in $S$.
This is entirely analogous to \Cref{claim:goal} of the warm-up result given in \Cref{sec:warmup}, and an argument entirely similar to the proof of \Cref{lem:weak-learn-over-sample} given \Cref{claim:goal} gives us the following analogue of \Cref{lem:weak-learn-over-sample}, saying that we can weak learn over a fixed sample of data points:

\begin{lemma} [Achieving non-trivial accuracy on a fixed sample] \label{lem:weak-learn-over-sample-functions}
Suppose that the input data set $S$ for \learnfunc\ is a sequence of $|S| = 2^{\Omega(\sqrt{n}\log^{\Omega(k)} (n))}$ examples that are labeled according to some 
function of $k$ halfspaces $f = g(\sgn(w^{(1)} \cdot x),\dots,\sgn(w^{(k)} \cdot x))$.  Then with probability at least $19/20$, \learnfunc\ outputs a hypothesis $h: \R^n \to \bits$ that correctly classifies at least
${\frac 1 2} + \gamma$  fraction of the examples in $S$, where $\gamma := 2^{-O(\sqrt{n} \log^{O(k)}(n))}$.
Moreover, the hypothesis class ${\cal H}$ of all hypotheses that can be generated by \learnfunc\ has VC dimension at most $\poly(n,\log(n)^k)$.
\end{lemma}

\begin{proof}
Since we repeat the inner loop in \learnfunc\ $2^{O(\sqrt{n} \log^{O(k)})}$ times, \Cref{lem:ltf-func-loop} implies that we output a hypothesis $h$ with advantage at least $\gamma$ with high probability. Moreover, note that the hypothesis output by \learnfunc\ is of the form $h_{b_1, b_2, b_3}$. In turn, we can write this as $g(f_0, \bigwedge_{i=1}^{\log^k(n)} f_i)$, where $f_0$ is an indicator of a subspace and each $f_i$ is an intersection of a degree $2$ PTF and a halfspace (as in the proof of \Cref{lem:weak-learn-over-sample}), corresponding to whether points lie in the region $R_{s_{t}}^{\beta_{r_t}}(g^{(t)})$. Note now that the VC dimension of the set of all linear subspaces is $n$ and the VC dimension of degree $2$ PTFs is at most $O(n^2)$. Thus, it follows directly from well-known techniques (see the proof of~Theorem~3.6 of \cite{KearnsVazirani:94}) that the VC dimension of $h$ is at most $\poly(n,\log^k(n))$, as desired.
\end{proof}

Given this, an argument entirely analogous to that given in \Cref{ap:proof-theorem-2} for the warm-up result gives a weak PAC learning algorithm, and a standard application of accuracy boosting as in \Cref{ap:proof-theorem-2} yields a strong PAC learner for the class of arbitrary functions of $k$ halfspaces.  This completes the proof of \Cref{thm:main}. 

\section*{Acknowledgements}
Josh Alman is supported in part by NSF Grant CCF-2238221 and a Packard Fellowship. 
Shyamal Patel is supported by NSF grants CCF-2106429, CCF-2107187, CCF-2218677, ONR grant ONR-13533312, and an NSF Graduate Student Fellowship. Rocco Servedio is supported by NSF grants CCF-2106429 and CCF-2211238.

\begin{flushleft}
\bibliographystyle{alpha}
\bibliography{allrefs}
\end{flushleft}

\appendix


\section{Proof of \Cref{thm:2} using \Cref{lem:weak-learn-over-sample}}
\label{ap:proof-theorem-2}

\subsection{Background on weak learning, strong learning, and boosting} \label{sec:PAC}

We recall standard definitions from computational learning theory and standard results from the theory of hypothesis boosting.

\medskip

\noindent {\bf (Distribution-free) PAC learning.}
Given a target Boolean function $f: \R^n \to \bits$,  a hypothesis $h: \R^n \to \bits$ and a distribution $\calD$ over $\R^n$, we say that $h$ is an {\em $\epsilon$-approximator for $f$ under ${\cal D}$}
if $\Pr_{{\calD}}[f(x) = h(x) ] \geq 1 - \epsilon.$  
A (strong) \emph{PAC learning algorithm} for a class ${\cal C}$ of Boolean-valued functions has access to an {\em example oracle} $EX(c, {\cal D})$ which, when
invoked, provides a labeled example $(\bx, f(\bx))$ where
$\bx$ is drawn from a fixed, but unknown and arbitrary, distribution $\calD$ and $f \in {\cal C}$ is the unknown \emph{target function} that the algorithm is
trying to learn. 
An algorithm $L$ is a {\em PAC learning algorithm for ${\cal C}$} if the following condition holds: for any unknown $f \in {\cal C},$ any unknown distribution ${\cal D}$,and any
$0 < \epsilon, \delta < 1$, if $L$ is given
$\epsilon$ and $\delta$ and has access to $EX(c,{\cal D}),$ then with probability at least
$1-\delta$ algorithm $L$ outputs an
$\epsilon$-approximator for $f$ under ${\cal D}.$

\medskip

\noindent {\bf Weak learning.}
For $\gamma > 0$ (which may depend on $n$ and other parameters), we say that an algorithm $W$ is a {\em weak learning algorithm for ${\cal C}$
with advantage $\gamma$} if $W$ satisfies the following condition:
For any unknown target function $f \in {\cal C}$, for any unknown distribution ${\cal D},$ if $W$ is given access to $EX(f,{\cal D})$ then
with probability at least $9/10$, $W$
outputs a hypothesis $h$ such that
$\Pr_{{\cal D}}[f(x) = h(x)] \geq 1/2 + \gamma.$

\medskip

\noindent {\bf Accuracy boosting.}  Well known results in computational learning theory provide explicit \emph{boosting} algorithms which can be used to  automatically and efficiently upgrade any weak learning algorithm into a strong PAC learning algorithm.  The following theorem is established in \cite{Schapire:90,Freund:95}, and also follows from a host of other boosting algorithms in the literature:

\begin{theorem} [Accuracy boosting] \label{thm:boost}
Let ${\cal C}$ be a class of functions over $\R^n$ and let $W$ be a weak learning algorithm for ${\cal C}$ with advantage $\gamma$ which runs in time at most $T$ when given $EX(f,{\cal D})$ for any target function $f \in {\cal C}$ and any distribution ${\cal D}$.
Then there is a (strong) PAC learning algorithm $L$ for ${\cal C}$ (which is obtained by applying a boosting algorithm to $W$) with the following property: on input parameters $0<\eps,\delta<1$,
the running time of $L$ is polynomial in $1/\gamma,$
$1/\epsilon,$ $\log 1/\delta,$ and $T.$
\end{theorem}

\subsection{A Weak Learning Algorithm for Intersections of Two Halfspaces} \label{sec:weak-learning}

With \Cref{lem:weak-learn-over-sample} in hand, which says that $\warmup$ constructs a ``simple'' yet nontrivial-accuracy hypothesis on any fixed sample $S$ labeled according to an intersection of two halfspaces, it is straightforward  obtain a distribution-free weak learning algorithm for the class of intersections of two halfspaces:

\begin{lemma} \label{lem:weak-learner}
There is an algorithm \warmupwl\
that is a weak learning algorithm for the class ${\cal C}$ of intersections of two halfspaces over $\R^n$ with advantage $\gamma' = 2^{-{O}(\sqrt{n}\log n)}.$  For any target intersection of two halfspaces and any distribution $\calD$ over $\R^n$, the running time of \warmupwl\ is $2^{\tilde{O}(\sqrt{n})}$.
\end{lemma}
\begin{proof}
The lemma is an easy consequence of the ``Fundamental Theorem of Statistical Learning'' (see e.g.~part~(1) of Theorem~6.8 of \cite{SSBD14}) and \Cref{lem:weak-learn-over-sample}.  In more detail, the \warmupwl\ algorithm calls the $EX(f,{\cal D})$ oracle $O \pbra{{\frac {d}{\gamma'^2}}}$ times, where $d=O(n^2)$ is the VC dimension of the hypothesis class ${\cal H}$ and $\gamma' = \gamma/2$ is half of the ``advantage over the sample'' parameter $\gamma$ from \Cref{lem:weak-learn-over-sample}. 
The labeled examples that are obtained from these calls form the input sequence of examples $\bS$ that is then given to \warmup, and the hypothesis that \warmupwl\ returns is the hypothesis generated by \warmup.  

By the uniform convergence property (part~(1) of Theorem~6.8 of \cite{SSBD14}), with probability at least $19/20$ over the draw of  $\bS$, the empirical accuracy of every hypothesis $h \in {\cal H}$ on $\bS$ is within an additive $\pm \gamma' = \pm \gamma/2$ of the true accuracy of $h$ on $f$ w.r.t.~${\cal D}$; and by \Cref{lem:weak-learn-over-sample}, with probability at least $19/20$ over the execution of \warmup, that algorithm generates a hypothesis $h \in {\cal H}$ with empirical accuracy at least $1/2 + \gamma$ over the sample $\bS$.  So with overall probability at least $9/10$, \warmupwl\ generates a hypothesis with advantage $\gamma'$ as claimed.
\end{proof}
\subsection{Proof of \Cref{thm:2}} \label{sec:proof-theorem-2}

\Cref{thm:2} follows directly by applying \Cref{thm:boost} (accuracy boosting) to \Cref{lem:weak-learner} (the weak learner for intersections of two halfspaces). \qed

\section{Brute-force search learns functions of $k$ halfspaces if $k$ is large} \label{ap:brute-force}

In this brief section we argue that, as claimed in \Cref{rem:large-k}, there is a PAC learning algorithm for the class of all functions of $k$ halfspaces over $\R^n$ that runs in time $\poly(2^{2^k + O(nk^2)}, 1/\eps,\log(1/\delta))$. 

It is a well-known consequence of standard techniques (see the proof of~Theorem~3.6 of \cite{KearnsVazirani:94}) that the VC dimension of the class of all Boolean functions of $k$ halfspaces over $\R^n$ is at most $O(2^k + nk \log k)$.
Moreover, as is well known, given a finite sample $S$ of points in $\R^n$, for any halfspace $h$ the value of $h$ on all points in $S$ can be encoded by giving the value of $h$ on $n+1$ (carefully chosen) points in $S$.  Given these facts, consider the following algorithm which operates on a sample $S \subset \R^n$ of $m' = O(2^k + nk \log k)$ points labeled according to some function of $k$ halfspaces:

\begin{itemize}
    \item For $i=1,\dots,k$, guess $n+1$ points in $S$ and use them to define a halfspace $h^{(i)}$ as alluded to above;
    \item Given $k$ halfspaces $h^{(1)},\dots,h^{(k)}$ from the preceding step, guess a function $g: \bits^k \to \bits$ and check whether $g(h^{(1)}(x),\dots,h^{(k)})$ correctly labels each point in $S$. Output a hypothesis $g(h^{(1)}(x),\dots,h^{(k)}(x))$ for which this is the case.
\end{itemize}

Since there are ${m' \choose n+1}$ ways to guess $n$ points for each $i$ in the first step, the total number of possibilities for the first step is at most ${m' \choose n+1}^k \leq (em'/(n+1))^{(n+1)k} =2^{O(nk^2)}.$
Since there are $2^{2^k}$ functions $g: \bits^k \to \bits$, the total number of possibilities for the second step is at most $2^{2^k}.$ Hence the above algorithm runs in time $2^{2^k + O(nk^2)}$ and is guaranteed to output a function of $k$ halfspaces that is consistent with the input data set of $m'$ points.
By the ``Fundamental Theorem of Statistical Learning'' (see e.g.~Theorem~6.8 of \cite{SSBD14}), for a suitable choice of the hidden constant in the definition of $m'$, the above algorithm, when run on a sample of $m'$ points drawn from $EX(c,{\cal D})$, is a weak PAC learning algorithm achieving error rate $\eps' = 0.01$ with probability at least $9/10.$  We can run a boosting algorithm on this weak PAC learning algorithm as described in \Cref{thm:boost}, and we get a strong PAC learning algorithm (achieving accuracy $\eps$ with probability $1-\delta$) running in the claimed time bound.

\end{document}